\documentclass[12pt, draftclsnofoot, onecolumn]{IEEEtran}
\usepackage{epsfig,latexsym}
\usepackage{float}
\usepackage{indentfirst}
\usepackage{amsmath}
\usepackage{amssymb}
\usepackage{times}
\usepackage{subfigure}

\usepackage{algorithm}
\usepackage[noend]{algpseudocode}

\usepackage{psfrag}
\usepackage{hyperref}
\usepackage{cite}
\usepackage{lastpage}
\usepackage{fancyhdr}
\usepackage{color}
 \usepackage{amsthm}
\usepackage{bigints}
\newtheorem{Proposition}{Proposition}
\newtheorem{Lemma}{Lemma}
\newtheorem{lemma}[Lemma]{$\mathbf{Lemma}$}
\newtheorem{proposition}[Proposition]{Proposition}

\newcounter{problem}
\newcounter{save@equation}
\newcounter{save@problem}
\makeatletter
\newenvironment{problem}
{\setcounter{problem}{\value{save@problem}}%
  \setcounter{save@equation}{\value{equation}}%
  \let\c@equation\c@problem
  \subequations
}
{\endsubequations
  \setcounter{save@problem}{\value{equation}}%
  \setcounter{equation}{\value{save@equation}}%
}
 
\begin{document}
\title{ No-Pain No-Gain: DRL Assisted Optimization in Energy-Constrained CR-NOMA Networks }\vspace{-0.25em}

\author{ Zhiguo Ding, \IEEEmembership{Fellow, IEEE}, Robert Schober, \IEEEmembership{Fellow, IEEE}, and H. Vincent Poor, \IEEEmembership{Life Fellow, IEEE}    \thanks{ 
  
\vspace{-2em}

    Z. Ding and H. V. Poor are  with the Department of
Electrical Engineering, Princeton University, Princeton, NJ 08544,
USA. Z. Ding
 is also  with the School of
Electrical and Electronic Engineering, the University of Manchester, Manchester, UK (email: \href{mailto:zhiguo.ding@manchester.ac.uk}{zhiguo.ding@manchester.ac.uk}, \href{mailto:poor@princeton.edu}{poor@princeton.edu}).
R. Schober is with the Institute for Digital Communications,
Friedrich-Alexander-University Erlangen-Nurnberg (FAU), Germany (email: \href{mailto:robert.schober@fau.de}{robert.schober@fau.de}).

  }\vspace{-2em}}
 \maketitle
 
\begin{abstract}  
This paper applies machine learning  to optimize the transmission policy of cognitive radio inspired    non-orthogonal multiple access (CR-NOMA) networks, where    time-division multiple access (TDMA) is used to serve  multiple primary users and   an energy-constrained secondary user is admitted to the primary users' time slots via NOMA. During   each time slot, the secondary user performs  the two tasks:  data transmission and   energy harvesting based on the signals received from the primary users. The goal of the paper is to   maximize the secondary user's long-term throughput, by optimizing its transmit power and   the time-sharing coefficient for its two tasks. The long-term throughput maximization problem is challenging   due to the need for making decisions that yield long-term gains but might result in short-term losses.  For example, when in a given time slot, a primary user with large channel gains transmits,   intuition suggests  that the secondary user should not carry out data transmission due to the strong interference from the primary user but perform energy harvesting only, which results in zero data rate for this time slot but yields potential long-term  benefits. In this paper, a deep reinforcement learning (DRL)  approach  is applied to emulate  this   intuition, where the deep deterministic policy gradient (DDPG) algorithm is employed together   with convex optimization. Our simulation results demonstrate that the proposed DRL assisted NOMA transmission scheme can yield significant performance gains over two benchmark schemes. 
\end{abstract} \vspace{-1em} 

\section{Introduction} 
Machine learning has been recognized as  one of the most   important enabling technologies for the next generation wireless networks \cite{7792374}. The key idea behind  machine learning is to learn to make optimal decisions based on observed   data or the environment \cite{MLABishop}.    Because of its general utility, machine learning has been   applied to a variety of  wireless communication problems. For example,  supervised learning has been applied to link adaption and channel estimation in orthogonal frequency-division multiplexing (OFDM) systems  \cite{5200378, 8052521}. Unsupervised learning has been shown to be particularly  beneficial for improving the accuracy of wireless positioning in \cite{7349230} and reducing  the complexity of     beamformer design   \cite{8586870}. Reinforcement learning  has also been shown to be   applicable to various communication problems, including energy harvesting, resource allocation, data and computation offloading, and network security \cite{8714026, 9108195, 6485022}. 

This paper   considers  the application of machine learning to non-orthogonal multiple access (NOMA) systems. We note that this application has already received significant attention in the literature \cite{mojobabook}. For example,   an unsupervised machine learning algorithm, K-means, has been applied to millimeter-wave NOMA networks for the joint design of beamforming and user clustering \cite{8454272}. A similar K-means based approach has also been proposed for terahertz NOMA systems employing  nodes equipped with multiple antennas \cite{9115278}.  A power minimization problem for multi-carrier NOMA facilitated by simultaneous wireless information and power transfer (SWIPT) has been solved by applying deep learning in \cite{8626195}. In \cite{9184017}, an intelligent offloading scheme powered  by deep learning was proposed for NOMA assisted mobile edge computing (NOMA-MEC). Furthermore, deep learning approaches have been developed for   the design of NOMA transceivers and the codebook of sparse code multiple access (SCMA)   in \cite{8387468} and \cite{8254356}, respectively. 
An overview of additional  related work on  the application of machine learning to NOMA can be found in recent survey articles \cite{9154358,8951180,8792153,Heruisi}.


This paper focuses on  a cognitive radio (CR) inspired NOMA network, where  multiple primary users and one energy-constrained secondary user communicate with the same base station, as shown in Fig. \ref{fig 0 ax}. Time-division multiple access (TDMA) is used to serve  the primary users, where the use of CR-NOMA ensures that the secondary user can be admitted to the primary users'  time slots without degrading their quality of service (QoS) experience  \cite{Zhiguo_CRconoma}. During   each time slot, the secondary user performs the   two tasks: Data transmission and   energy harvesting from  the signals sent by the primary users  \cite{Zhouzhang13}. 
We note that even if   all the users' channels are constant, for the secondary user, its communication environment is time-varying and changes from one time slot to the next, since different primary users are scheduled in different time slots, as shown in  Fig. \ref{fig 0 bx}. This time-varying pattern is difficult to   handle with conventional optimization tools, but can be effectively learned by machine learning.
 
The objective of this paper is to maximize  the secondary user's long-term throughput,
 by optimizing its transmit power and also the time-sharing coefficient for its two tasks. This throughput maximization problem  is challenging because of the need for making    decisions which  might yield a long-term gain but could result in a short-term loss. For example, when a primary user with strong channels   to both the base station and the secondary user transmits in a given time slot, a human would decide  that the secondary user should not transmit data but  carry  out energy harvesting only, since the primary user would cause severe interference. Even though this decision leads to zero data rate in this particular time slot, it may yield  a   performance gain in the future since a significant amount of energy can be harvested due to the primary user's strong channel conditions.    With  conventional optimization tools, it is difficult to emulate this `no-pain no-gain'-like   decision,  which motivates the application of machine learning to the   considered NOMA scenario, because  machine learning algorithms are well known for their    capability of     emulating human behaviour and making `no-pain no-gain'-like   decisions \cite{MLABishop,suttonrl,rlsurvey}.

 {Recall that compared to supervised and unsupervised machine learning, reinforcement learning has the  capability to make decisions and perform learning at the same time, which is a desirable  feature for wireless communication applications  \cite{8714026,suttonrl,9023459,8500158}. Deep deterministic policy gradient (DDPG) is an effective implementation of      deep reinforcement learning (DRL), where the  action-value functions are approximated by neural networks and the  action space does not have to be discrete  \cite{drlX,ddpg}.  Because   the optimization variables for the considered long-term throughput maximization problem   are continuous,      DDPG is adopted  in this paper.}   In order to facilitate the application of DDPG, the considered long-term throughput maximization problem is first decomposed into   two optimization subproblems. The first subproblem  focuses on a single time slot and optimizes the secondary user's transmit power and the time sharing coefficient, where the closed-form optimal solution  is found by applying convex optimization. The second subproblem targets the optimization of  the energy fluctuation across  different time slots, and is solved by applying the adopted DRL  tool, DDPG.  Extensive simulation results are provided   to demonstrate that the proposed DDPG assisted NOMA scheme   effectively emulates human decisions and offers significant performance gains over two benchmark schemes.

\section{System Model}\label{section model}
 \begin{figure}[!t]\vspace{-1em}
\begin{center} \subfigure[ System Model ]{\label{fig 0 ax}\includegraphics[width=0.52\textwidth]{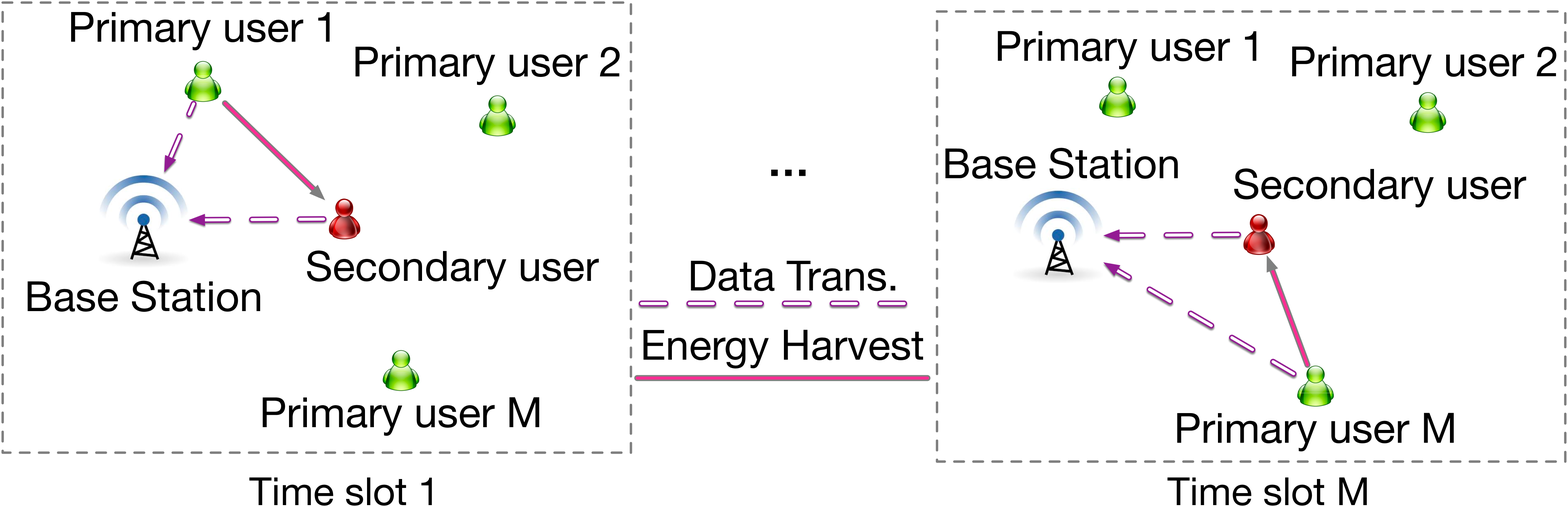}}
\subfigure[CR-NOMA Transmission ]{\label{fig 0 bx}\includegraphics[width=0.42\textwidth]{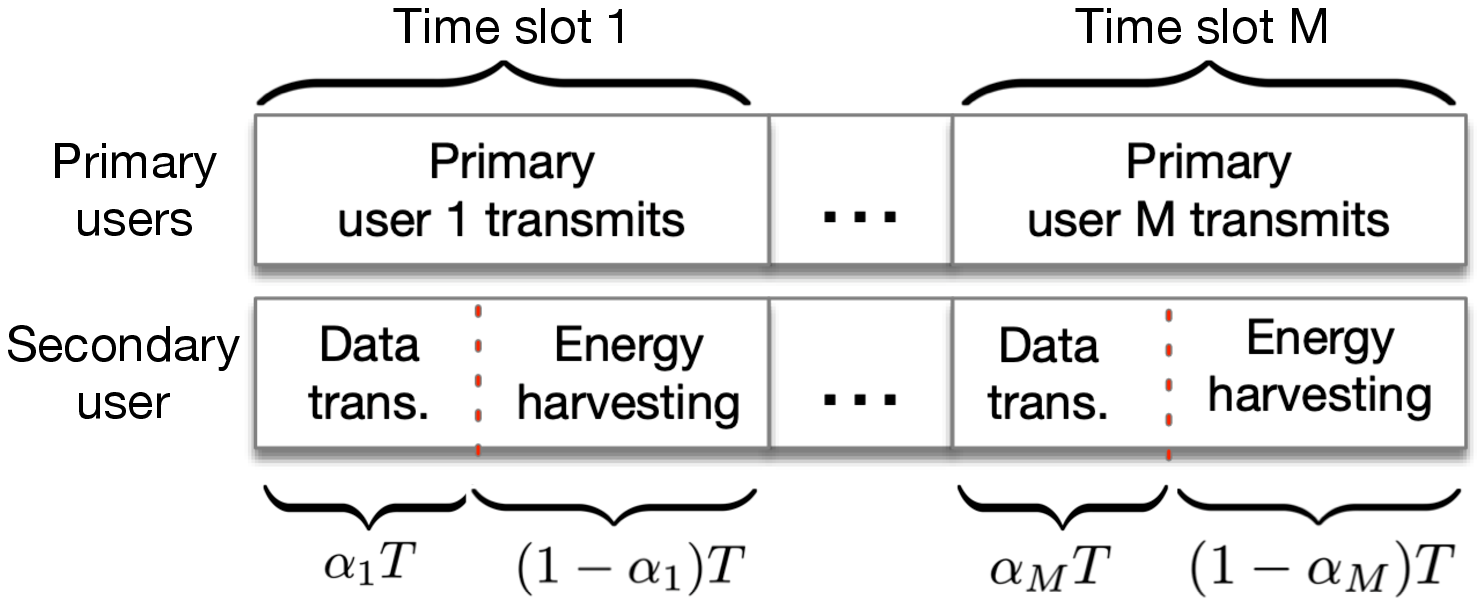}} \vspace{-1em}
\end{center}
\caption{   System diagram for the considered CR-NOMA uplink communication scenario, where a secondary  user shares the spectrum with $M$ primary  users and harvests energy from the signals sent by the primary  users.    }\label{fig0}\vspace{-2em}
\end{figure}

Consider a   CR-NOMA uplink  scenario as shown in Fig. \ref{fig 0 ax}, with   one base station, $M$   primary users, denoted by ${\rm U}_m$, $1\leq m \leq M$, and one energy-constrained secondary user, denoted by ${\rm U}_0$. 
 {It is assumed that the  $M$ primary users are scheduled to transmit in different   time slots as shown in Fig. \ref{fig 0 bx}, and this TDMA based   scheduling is fixed and carried out over a long period of time such that each primary user is scheduled to transmit multiple times\footnote{ {We note that    the key idea of the proposed NOMA scheme is to ensure that the admission of the secondary user as well as the proposed  resource allocation for the secondary user are transparent to the primary users. In other words, we do not adjust  the primary users'   transmission strategies because it is not the primary users' responsibility  to accommodate the secondary user.}}.}    In particular, in the $n$-th time slot, denoted by $t_n$, user ${\rm U}_m$,  $m=(n-1)\oplus M+1$, is scheduled to transmit, where $\oplus$ denotes the modulo operation. For example, for $M=2$, the primary users scheduled at $t_1$, $t_2$, $t_3$, and $t_4$ are ${\rm U}_1$, ${\rm U}_2$, ${\rm U}_1$, and ${\rm U}_2$, respectively.  {For notational convenience, denote the primary user scheduled at $t_n$ by $\bar{\rm U}_n$, i.e., $\bar{\rm U}_n= {\rm U}_m$, $m=(n-1)\oplus M+1$.   Denote the channel gain between the base station and $\bar{\rm U}_n$ by $h_n$,   the channel between  ${\rm U}_0$ and  $\bar{\rm U}_n$ by $h_{n,0}$, and the channel between ${\rm U}_0$ and  the base station at $t_n$ by $g_{0,n}$,     where  the effects of  both large-scale path loss and quasi-static multi-path fading are included.  } 


The secondary user is admitted to the primary users' time slots via CR-NOMA. In particular, in  time slot $t_n$,  the secondary user uses the first $\alpha_n T$ seconds   for its data transmission, where $\alpha_n$ is a time-sharing parameter, $0\leq \alpha_n\leq 1$, and $T$ denotes the duration of each time slot. The remainder of the time slot, $(1-\alpha_n)T$ seconds,  will be used for battery charging by  harvesting energy from the signals sent by primary user $\bar{\rm U}_n$, as shown  in Fig. \ref{fig 0 bx}. 
At the beginning of $t_n$, the scheduled primary  user's channel state information (CSI), i.e., $h_n$ and $h_{n,0},$ is assumed to be available at   the secondary user\footnote{ {The required CSI knowledge can be realized as follows. Prior to its transmission,    primary user $\bar{\rm U}_n$ can broadcast a pilot signal which allows  the base station to estimate $h_n$ and the secondary user to estimate $h_{n,0}$ simultaneously. Via a reliable feedback link,  the base station can pass its knowledge of   $h_n$ to the secondary user. The base station needs to further broadcast a pilot signal and   to enable  the secondary user to estimate $g_{0,n}$.   For low-mobility applications, such as  Internet of Things with static sensors, the system overhead caused by channel estimation  is moderate since the pilot signals can be sent  infrequently.   }}. 

Denote   the remaining energy in the secondary  user's battery at the beginning of the $n$-th time slot   by $E_n$, which means that $E_{n+1}$ can be expressed as follows:
\begin{align} 
&E_{n+1} = \min\left\{ \underset{\text{ Harvested energy  }}{\underbrace{(1- \alpha_n) T \eta P_n|h_{n,0}|^2}}-\underset{\text{ Used energy  }}{\underbrace{ \alpha_n TP_{0,n}}}  +E_n, E_{\max}\right\}, \label{E_{n+1}}
\end{align}
where $E_{\max}$ denotes the capacity limit of the secondary user's battery, $\eta$ denotes the energy harvesting efficiency coefficient,  $P_n$  and $P_{0,n}$ denote the transmit powers of $\bar{\rm U}_n$ and ${\rm U}_0$ at $t_n$, respectively.  {Assume that the initial amount of energy available in the secondary user's battery is  $E_{\max}$, i.e.,  $E_1=E_{\max}$. We note that  \eqref{E_{n+1}} ensures that the amount of   harvested  energy   cannot exceed   the battery capacity, $E_{\max}$.}

Due to the energy  causality  constraint,    the secondary user's transmit power is capped by the energy available in its battery, i.e.,
\begin{align}\label{power constraint}
\alpha_n TP_{0,n}\leq E_n  .
\end{align}
We note that \eqref{power constraint}   ensures that $E_{n+1}$ shown in \eqref{E_{n+1}} is always non-negative. 

Therefore, the data rate that the secondary user can achieve  at $t_n$ is given by\footnote{For notational convenience, the noise power is assumed to be normalized, such that the value of   the noise power is absorbed into the channel gains $|g_{0,n}|^2$ and $|h_n|^2$. }
\begin{align}\label{raten}
R_n = \alpha_n \ln\left( 1+\frac{P_{0,n}|g_{0,n}|^2}{1 + P_n|h_n|^2}
\right), \text{ [nats per channel use (NPCU)]} ,
\end{align}
 {which guarantees that the secondary user's signal can   be successfully decoded during the first stage of successive interference cancellation (SIC) at the base station. After the base station removes the successfully decoded signal of the secondary user from the received signal, the primary users' signals can be decoded in the same manner   as when the secondary user is absent. In other words,  \eqref{raten}  guarantees  that the primary users' QoS experience is not   affected by admitting the secondary user into their  time slots.} We note that \eqref{raten}  is based on      the  QoS-based SIC decoding order, but more sophisticated  SIC strategies can be used to further improve the performance of the CR-NOMA scheme, as discussed in \cite{hsic01}. 
 
 The aim of this paper is to maximize the secondary user's long-term  throughput by optimizing the time-sharing parameter, $\alpha_n$, and  {the secondary user's transmit power, $P_{0,n}$, in the $n$-th time slot $t_n$. Define  the secondary user's transmission environment   at $t_n$, also termed the state in the context of reinforcement learning, as $s_n$. One choice of $ {s}_n$ can be  $ {s}_n=\begin{bmatrix}|g_{0,n}|^2& |h_n|^2& |h_{n,0}|^2  & E_n\end{bmatrix}^T$, where $[\cdot]^T$ denotes a transpose operation.  Define  the policy which is a strategy, i.e., a set of sequential actions $\begin{bmatrix}P_{0,1}, \alpha_1, P_{0,2}, \alpha_2, \ldots \end{bmatrix}$,   adopted  by the secondary user as  $\pi_{P_{0,n}, \alpha_n }$.   Then,  the corresponding   optimization problem can be formulated as follows \cite{suttonrl,9108195, 6485022}: 
   \begin{problem}\label{pb:1}
  \begin{alignat}{2}
\underset{\pi_{P_{0,n}, \alpha_n } }{\rm{max}} &\quad   
\mathcal{E} \left\{ \left. \sum^{\infty}_{n=1}\gamma^{n-1} R_n \right| \pi_{P_{0,n}, \alpha_n },  {s}_1 \right\}
\label{obj:1} \\
\rm{s.t.} & \quad  
 E_{n+1} = \min\left\{ (1-\alpha_n) T \eta P_n|h_{n,0}|^2 -\alpha_n TP_{0,n} +E_n,E_{\max}\right\}
\label{1st:1}
\\&\quad 
\alpha_nTP_{0,n}\leq   E_n  \label{2st:1}
\\
& \quad 
0\leq \alpha_n\leq 1
\label{3st:1} 
\\
& \quad 
0\leq P_{0,n}\leq  P_{\max},
\label{4st:1} 
  \end{alignat}
\end{problem}  
 where $\mathcal{E} \left\{ \left. \sum^{\infty}_{n=1}\gamma^{n-1} R_n \right| \pi_{P_{0,n}, \alpha_n },  {s}_1 \right\}$ denotes the expected value of the discounted cumulative  sum of the data rates for given $\pi_{P_{0,n}, \alpha_n }$ and $ {s}_1$, also termed the state-value function for policy $\pi_{P_{0,n}, \alpha_n }$ in the context of reinforcement learning\footnote{ {The  expectation carried out in \eqref{obj:1} is explained in the following. In particular, given the policy, the secondary user's action at $t_n$, $n\geq 1$, is fixed. However,  even though both $ {s}_1$ and the actions are fixed, the realization of the  state   and  the secondary user's data rate at $t_n$, $n>1$, are not deterministic  due to the time-varying environment. Hence,  the expectation is needed to capture the statistical property of the state transitions.    }},   $P_{\max}$ denotes the maximal transmit power of the secondary user, and $\gamma$ denotes the  discount rate parameter \cite{suttonrl}. We note that the use of the discount rate parameter ensures that   a decision which yields a long-term gain is preferred over a short-sighted decision which maximizes the instantaneous data rate only (i.e., $\gamma=0$). In other words, by varying the choice of $\gamma$ between $0$ and $1$, a different tradeoff between the long-term gain and instantaneous benefit can be achieved. We note that $\gamma$ has to be strictly smaller than one, which is to ensure  that the infinite sum in the objective function of \ref{pb:1} is   finite.  }

{\it Remark 1:} One can also formulate the problem by asking the secondary user to charge its battery in each time slot before transmitting its data. Intuitively, both formulations should lead to the same   problem formulation, which however  is not true. In particular, if the secondary user carries out energy harvesting first, the energy available for data transmission is given by 
\begin{align}
\min\{E_{\max}, E_n + \alpha_n T \eta P_n|h_{n,0}|^2\} , 
\end{align}
which means that the constraints in \eqref{2st:1} and \eqref{3st:1} need to be formulated as follows:
\begin{align}\label{case5}
 E_{n+1} =& \max\left\{ \min\{E_{\max}, E_n + \alpha_n T \eta P_n|h_{n,0}|^2\}  -(1-\alpha_n)TP_{0,n},0 \right\}, \\ 
 \alpha_nTP_{0,n}\leq&  \min\{E_{\max}, E_n + \alpha_n T \eta P_n|h_{n,0}|^2\},\label{case6}
\end{align}
respectively. 
The constraints in \eqref{case5} and \eqref{case6} are more complicated than   \eqref{2st:1} and \eqref{3st:1}, which makes solving  the rate maximization problem more difficult.  

{\it Remark 2:} The  following two strategies will be  used as  benchmark schemes in this paper. For the first scheme,    termed the greedy algorithm,   ${\rm U}_0$ uses all the energy available for data transmission, and then starts energy harvesting. In particular, ${\rm U}_0$'s transmit power is fixed at $P_{\max}$, and $\alpha_n=\min\left\{1,\frac{E_n}{TP_{\max}}\right\}$ is adopted by the greedy algorithm. For the case that $TP_{\max}\geq E_n$, this choice of $\alpha_n$ means that all the energy available, $E_n$, will be used to power data transmission, i.e., $TP_{\max}\alpha_n=E_n$. For the second    benchmark scheme,    termed the random algorithm,     $P_{\max}$ is used as ${\rm U}_0$'s transmit power and $\alpha_n$ is uniformly   generated between $0$ and $\min\left\{1,\frac{E_n}{TP_{\max}}\right\}$. Because  $\alpha_n\leq\frac{E_n}{TP_{\max}}$, it is guaranteed that there is enough energy for ${\rm U}_0$ to transmit during $\alpha_nT$ seconds with transmit power $P_{\max}$.

{\it Remark 3:} Problem \ref{pb:1} cannot be straightforwardly  solved by applying conventional convex optimization, not only because its equality constraint in \eqref{1st:1} is not affine, but also because its objective function is a long-term non-convex   throughput function. However, Problem \ref{pb:1} is   ideally suited  for the application of machine learning, as discussed in the following section.   

\section{Problem Reformulation for Application of Reinforcement Learning} \label{section III}
\subsection{Rationale Behind the Application of  Reinforcement Learning}
 {For the considered optimization problem,   `no-pain no-gain'   decisions, i.e., decisions which  result in  short-term losses  but   yield long-term gains, have to be made. These decisions motivate the application of reinforcement learning.     In the following, a simple   example with two primary users,  i.e., ${\rm U}_1$ and ${\rm U}_2$,  is used to illustrate the `no-pain no-gain'  situations  inherent to the considered CR-NOMA scenario. }

In particular, assume that  ${\rm U}_1$ has  very strong channels to both ${\rm U}_0$ and the base station, i.e., $|h_{1}|^2\rightarrow \infty$ and $|h_{1,0}|^2\rightarrow \infty$,  whereas ${\rm U}_2$ has very weak  channels to both ${\rm U}_0$ and the base station, i.e., $|h_{2}|^2\rightarrow 0$ and $|h_{2,0}|^2\rightarrow 0$.  Intuition in  this situation suggests  to encourage  the secondary user to do the following:
\begin{itemize}
\item Use   most of the time slot for energy harvesting when  ${\rm U}_1$ transmits. The reason is that $|h_{1}|^2\rightarrow \infty$ means that ${\rm U}_1$ is a strong interference source for ${\rm U}_0$'s data transmission,  and $|h_{1,0}|^2\rightarrow \infty$ means an   opportunity to harvest a large amount of energy from ${\rm U}_1$.

\item Use   most of the time slot for data transmission when ${\rm U}_2$ transmits. The reason is that  $|h_{2}|^2\rightarrow 0$ implies   an interference free situation, and hence it is possible to achieve  a high  data rate,  as indicated by \eqref{raten}.  $|h_{1,0}|^2\rightarrow 0$ implies that only a small  amount of energy can be harvested from  ${\rm U}_2$. 
\end{itemize}

Note that the actions following this   intuition  might incur    short-term losses compared to the benchmark schemes. For example, when ${\rm U}_1$ transmits, a rationale  action is to ask ${\rm U}_0$ to carry out energy harvesting only, which leads to  zero data rate     in this particular time slot.  In other words, in the time slots where ${\rm U}_1$ transmits, the data rate achieved by  the action following the   intuition  is smaller than the data rates  achieved by  the benchmark schemes. However, this temporary sacrifice yields a large amount of harvested energy  and hence a potential long-term gain.  
Reinforcement learning is an ideal    tool for emulating   this   intuition by learning  the  feature inherent  to the considered problem, as shown in the remainder of this paper.

\subsection{Problem Reformulation}
 {For the throughput maximization problem formulated  in \ref{pb:1},  there are two sets of optimization variables, namely $P_{0,n}$ and $\alpha_n$, and  the value ranges of $P_{0,n}$ and $\alpha_n$ can be quite   different, i.e., $0\leq P_{0,n}\leq P_{\max}$ and $0\leq \alpha_n \leq 1$, which make a direct application of DDPG challenging. In the following, problem \ref{pb:1} will  be decomposed into two simpler optimization subproblems, in order to facilitate the application of DDPG.}  First, similar to \cite{9108195}, we introduce an energy fluctuation parameter, which is defined as the difference between the energy consumed and the energy harvested at $t_n$: 
\begin{align}
\bar{E}_n = (1-\alpha_n) T \eta P_n|h_{n,0}|^2 - \alpha_n TP_{0,n}.
\end{align}
$\bar{E}_n$ can be interpreted as the energy surplus (or deficit)  at $t_n$ if $\bar{E}_n>0$ (or  $\bar{E}_n<0$).  We note that for a fixed $\bar{E}_n $,     the data rate at $t_n$, $R_n$,   depends on the power allocation coefficient and the time-sharing parameter at $t_n$ only, where the parameters for the other  time slots, i.e., $\alpha_m$ and $P_{0,m}$, $m\neq n$, will not have any impact on $R_n$.  This observation can be clearly illustrated by recasting problem \ref{pb:1} in the following equivalent form \cite[Page 133]{Boyd}: {
 \begin{problem}\label{pb:1x}
  \begin{alignat}{2}
\underset{ \pi_{ \bar{E}_n}}{\rm{max}} &\quad   
\mathcal{E} \left\{ \left. \sum^{\infty}_{n=1}\gamma^{n-1} R_n \right| \pi_{ \bar{E}_n},  {s}_1 \right\}
\label{obj:1x} \\
\rm{s.t.} & \quad  
 E_{n+1} = \min\left\{E_{\max}, E_n + \bar{E}_n\right\}  
\label{1st:1x} ,  \end{alignat}
\end{problem} 
where  the subscript  of   policy $ \pi_{ \bar{E}_n}$  highlights the fact that the secondary user's action at $t_n$ is to choose $ \bar{E}_n$ (instead of $P_{0,n}$), $\alpha_n$,} and  $\tilde{R}_n$ is  defined as follows:
\begin{align}\nonumber
\tilde{R}_n=&\sup \left\{  {R}_n|   (1-\alpha_n) T \eta P_n|h_{n,0}|^2 - \alpha_n TP_{0,n}=\bar{E}_n ,   
\eqref{2st:1} , \eqref{3st:1} , \eqref{4st:1}\right\}.
\end{align}

Therefore, problem \ref{pb:1} can be solved by first solving the following optimization problem:\footnote{ {For the    subproblem shown in \ref{pb:2}, only the variables       for $t_n$ are involved, i.e., the secondary user's data rate at $t_n$ is maximized by optimizing  its transmit power   $P_{0,n}$ and its time-sharing parameter $\alpha_n$ at $t_n$. Therefore, only the causal CSI assumption is require, and the   optimal solution of problem \ref{pb:2} is  applicable regardless of whether the channels are  time-varying or contant.}}
 \begin{problem}\label{pb:2}
  \begin{alignat}{2}
\underset{P_{0,n}, \alpha_n}{\rm{max}} &\quad   
  R_n 
\label{obj:2} \\
\rm{s.t.} & \quad   
   (1-\alpha_n) T \eta P_n|h_{n,0}|^2 - \alpha_n TP_{0,n}=\bar{E}_n \label{1st:2}
\\\nonumber
& \quad 
\eqref{2st:1} , \eqref{3st:1} , \eqref{4st:1}.
  \end{alignat}
\end{problem} 
Denote    the optimal solutions of problem \ref{pb:2} by $\alpha_n^*(\bar{E}_n)$ and $P_{0,n}^*(\bar{E}_n)$, where $\alpha_n^*(\bar{E}_n)$ and $P_{0,n}^*(\bar{E}_n)$ are expressed as functions of $\bar{E}_n$ in order to highlight the fact that the optimal solutions,  $\alpha_n$ and $P_{0,n}$, are obtained for a given  $\bar{E}_n$. The closed-form expressions for $\alpha_n^*(\bar{E}_n)$ and $P_{0,n}^*(\bar{E}_n)$ are then substituted in problem \ref{pb:1x}, which yields the following optimization problem:\footnote{ {We note that \ref{pb:3} will be solved by applying   DRL which is well known for its applicability  in time-varying environments. As a result, the proposed DRL scheme can be directly applied if the channels are   time-varying,  as demonstrated in Section \ref{section VI}. Furthermore, we note that   the use of DRL to solve \ref{pb:3} requires   causal CSI   only, since the principle of DRL is to generate a decision by sensing the current state of the environment and using the past experience, where the knowledge  of the environments in the future is not needed.  }} {
 \begin{problem}\label{pb:3}
  \begin{alignat}{2}
\underset{\pi_{ \bar{E}_n}}{\rm{max}} &\quad    
 \mathcal{E} \left\{ \left. \sum^{\infty}_{n=1}\gamma^{n-1} R_n\left(\alpha_n^*(\bar{E}_n),P_{0,n}^*(\bar{E}_n)\right)   \right| \pi_{ \bar{E}_n},  {s}_1 \right\}
\label{obj:3} \\
\rm{s.t.} & \quad  E_{n+1} = \min\left\{E_{\max}, E_n + \bar{E}_n\right\}
\label{1st:3}
.
  \end{alignat}
\end{problem}} 
where $R_n\left(\alpha_n^*(\bar{E}_n),P_{0,n}^*(\bar{E}_n)\right)$ is a function of $\bar{E}_n$ only, i.e., $R_n\left(\alpha_n^*(\bar{E}_n),P_{0,n}^*(\bar{E}_n)\right)=\alpha^*(\bar{E}_n) \ln\left( 1+\frac{P_{0,n}^*(\bar{E}_n)|g_{0,n}|^2}{1 + P_n|h_n|^2}
\right)$. 

In Section \ref{section IV}, problem \ref{pb:2} will be solved via convex optimization, where  closed-form expressions for $\alpha_n^*(\bar{E}_n)$ and $P_{0,n}^*(\bar{E}_n)$ will be presented. In Section \ref{section V}, problem \ref{pb:3} will be reformulated as a reinforcement learning problem and the DDPG algorithm will be applied to find the solution. 



\section{Finding Closed-Form Expressions for $\alpha_n^*(\bar{E}_n)$ and $P_{0,n}^*(\bar{E}_n)$} \label{section IV}

%

In order to find   closed-form expressions for  $\alpha_n^*(\bar{E}_n)$ and $P_{0,n}^*(\bar{E}_n)$, we first   recast  \eqref{pb:2} as follows:
 \begin{problem}\label{pb:35}
  \begin{alignat}{2}
\underset{P_{0,n}, \alpha_n}{\rm{max}} &\quad   f_0\left(P_{0,n}, \alpha_n\right) \triangleq 
  \alpha_n \ln\left( 1+\frac{P_{0,n}|g_{0,n}|^2}{1 + P_n|h_n|^2}
\right)
\label{obj:35} \\
\rm{s.t.} & \quad   \label{1st:35}
 f_1\left(P_{0,n}, \alpha_n\right) = 0, 
 \\ 
&  \quad 
 f_2\left(P_{0,n}, \alpha_n\right) \leq 0, 
\\ 
&  \quad 
\eqref{3st:1}, \eqref{4st:1}.
  \end{alignat}
\end{problem} 
where 
$f_1\left(P_{0,n}, \alpha_n\right) =  (1- \alpha_n) T \eta P_n|h_{n,0}|^2 - \alpha_nTP_{0,n}- \bar{E}_n$ and $f_2\left(P_{0,n}, \alpha_n\right)  = \alpha_nTP_{0,n}-   E_n $.

{\it Remark 4:} Note that problem \ref{pb:35} is not   jointly convex in $P_{0,n}$ and $\alpha_n$. For example, its equality constraint, \eqref{1st:35}, is not an affine function. A similar optimization problem was considered  in \cite{9108195}, where the equality constraint was relaxed to an inequality  constraint. However, in the expression for the relaxed inequality constraint shown in \cite[Eq. (13)]{9108195}, there is still a term which involves  the multiplication of two optimization variables, which means that the inequality function is not convex and hence the relaxed problem is not in a convex form.

According \cite[Page 133]{Boyd},  an equivalent form of problem \eqref{pb:35} can be found as follows:
 \begin{problem}\label{pb:4}
  \begin{alignat}{2}
\underset{  \alpha_n}{\rm{max}} &\quad     \tilde{f}_0\left( \alpha_n\right) 
\label{obj:4} \\
\rm{s.t.} &  \quad 
0\leq \alpha_n\leq 1 ,
  \end{alignat}
\end{problem} 
where $ \tilde{f}_0\left( \alpha_n\right) $ is defined as follows:
\begin{align}
 \tilde{f}_0\left( \alpha_n\right)  = \sup &\left\{ f_0\left(P_{0,n}, \alpha_n\right)|  f_1\left(P_{0,n}, \alpha_n\right) = 0,   f_2\left(P_{0,n}, \alpha_n\right) \leq  0, P_{0,n}\geq 0 \right\}.
\end{align}
Finding $ \tilde{f}_0\left( \alpha_n\right)$ is equivalent to solving the following optimization problem:
 \begin{problem}\label{pb:5}
  \begin{alignat}{2}
\underset{ P_{0,n}}{\rm{max}} &\quad     \alpha_n \ln\left( 1+\frac{P_{0,n}|g_{0,n}|^2}{1 + P_n|h_n|^2}
\right)
\label{obj:5} \\
\rm{s.t.} &  \quad 
  P_{0,n} = \frac{(1- \alpha_n) \eta P_n|h_{n,0}|^2}{\alpha_n} -  \frac{\bar{E}_n}{\alpha_nT}
 \label{1st:5}
 \\
  & \quad  P_{0,n}\leq  \frac{ E_n}{ \alpha_nT}
  \label{2st:5}
   \\
  & \quad 0\leq P_{0,n}\leq  P_{\max}
  \label{3st:5}.
  \end{alignat}
\end{problem} 
Note that problem \ref{pb:5} is a function of $P_{0,n}$ only, where $\alpha_n$ is fixed. 
By using the equality constraint in \eqref{1st:5}, the optimal value of  $ \tilde{f}_0\left( \alpha_n\right) $ can be found as follows:
\begin{align}
 \tilde{f}^*_0\left( \alpha_n\right)  =   \alpha_n \ln\left( 1+\frac{  \frac{(1- \alpha_n) \eta P_n|h_{n,0}|^2|g_{0,n}|^2 }{\alpha_n}-  \frac{\bar{E}_n|g_{0,n}|^2 }{\alpha_nT} }{1 + P_n|h_n|^2}
\right),
\end{align}
where the constraints in \eqref{2st:5} and \eqref{3st:5} can be guaranteed by expressing  the domain of $ \tilde{f}_0\left( \alpha_n\right)$ as follows:
\begin{align}
\mathcal{D} = &\left\{\alpha_n| 0\leq    \frac{(1- \alpha_n) \eta P_n|h_{n,0}|^2}{\alpha_n} -  \frac{\bar{E}_n}{\alpha_nT}  \leq\min\left\{ P_{\max}, \frac{ E_n}{ \alpha_nT}\right\} \right\}.
\end{align}

By using $ \tilde{f}^*_0\left( \alpha_n\right) $ and showing   the constraints on the domain of the function explicitly, problem \eqref{pb:4} can be equivalently recast  as follows:
 \begin{problem}\label{pb:7}
  \begin{alignat}{2}
\underset{  \alpha_n}{\rm{max}} &\quad  \alpha_n \ln\left( 1+\frac{  \frac{(1- \alpha_n) \eta P_n|h_{n,0}|^2|g_{0,n}|^2 }{\alpha_n}-  \frac{\bar{E}_n|g_{0,n}|^2 }{\alpha_nT} }{1 + P_n|h_n|^2}
\right)
\label{obj:7} \\ 
\rm{s.t.} &  \quad 
   \alpha_n \leq 1-  \frac{\bar{E}_n}{ T\eta P_n|h_{n,0}|^2}  \label{1st:7} \\
  &  \quad 
  \alpha_n   \geq 1- \frac{ E_n+\bar{E}_n}{ T\eta P_n|h_{n,0}|^2} \label{2st:7} \\
    &  \quad 
\alpha_n\geq \frac{T\eta P_n|h_{n,0}|^2-  \bar{E}_n}{T\eta P_n|h_{n,0}|^2 + TP_{\max}}   \label{3st:7}\\
  &  \quad 
0\leq \alpha_n\leq 1 \label{4st:7}.
  \end{alignat}
\end{problem} 
Note that   problem \ref{pb:7} involves  three lower bounds and two upper bounds on $\alpha_n$. In addition, the domain of the objective function in \eqref{obj:7} also imposes a constraint on the choice of $\alpha_n$. Therefore, the problem might be infeasible  if one of the lower bounds is larger than one of the upper bounds, and hence it is important to carry out a feasibility study for problem \ref{pb:7}, as is done  in the following proposition.
\begin{proposition}\label{proposition1}
Problem \ref{pb:7} is always feasible.  
\end{proposition}
\begin{proof} 
See Appendix \ref{proof1}. 
\end{proof}

{\it Remark 5:} For the   case $\bar{E}_n=T\eta P_n|h_{n,0}|^2$, the upper bound in \eqref{1st:7} becomes $0$, which means that the only feasible solution for this case is $\alpha_n^*(\bar{E}_n)=0$. Therefore, no data transmission happens and $R_n=0$ for this case. We note that the case with   $\bar{E}_n=T\eta P_n|h_{n,0}|^2$ can cause  a singularity issue when analyzing the convexity of the objective function. Therefore, unless otherwise stated, it is assumed that $\bar{E}_n\neq T\eta P_n|h_{n,0}|^2$ for the remainder of this section. 

Another important step in finding  the optimal solution of  problem \ref{pb:7} is to study the convexity of its objective function,   which is done  in the following proposition.

\begin{proposition}\label{proposition2}
The objective function of problem \ref{pb:7} is a concave function of $\alpha_n$, for $\alpha_n\geq 0$. 
\end{proposition}

\begin{proof}
See Appendix \ref{proof2}.
\end{proof}

Given the two useful properties provided  in Propositions \ref{proposition1} and \ref{proposition2}, it can be easily verified that problem \ref{pb:7} is a concave problem; however, finding a closed-form solution for problem \ref{pb:7}  is challenging, due to the multiple constraints on the choices of $\alpha_n$ as shown in \eqref{1st:7}, \eqref{2st:7}, \eqref{3st:7},  and \eqref{4st:7}. It is important to point out that these constraints cannot be merged. For example,    the constraint in \eqref{2st:7} is not always  stricter than $\alpha_n\geq 0$ shown in  \eqref{4st:7} since
\begin{align}
 1- \frac{ E_n+\bar{E}_n}{ T\eta P_n|h_{n,0}|^2} &=   \frac{T\eta P_n|h_{n,0}|^2 -  E_n-\bar{E}_n}{ T\eta P_n|h_{n,0}|^2}  
\end{align} 
which can be either negative or positive.   As a result, finding the optimal solution by directly applying the Karush-Kuhn-Tucker (KKT) conditions to problem \ref{pb:7} is difficult, since  five dual variables are required due to the large number of inequality constraints.  Furthermore, the fact that the root of the first-order derivative of the objective function is in the form of the Lambert W function with two branches makes it more challenging to find  a closed-form expression for the optimal solution.  

Nevertheless, by using the properties of the inequality constraints of problem \ref{pb:7}, a closed-form optimal solution can be found, as shown in the following lemma. 

\begin{lemma}\label{lemma1}
The optimal solution for problem \ref{pb:7} can be expressed as follows:
\begin{align}\label{alphaoptimal}
\alpha_n^*(\bar{E}_n) =  \min\left\{1, \max\left\{ x^*, \theta_0\right\} \right\}, 
\end{align}
if $\bar{E}_n\neq T\eta P_n|h_{n,0}|^2$, otherwise $\alpha_n^*(\bar{E}_n) =0$,
where  
$x^*= \frac{\kappa_1-\kappa_2}{e^{ W_0\left(e^{-1}(\kappa_1-1)\right)+1}  - 1+\kappa_1}$,  $\theta_0=\max\left\{1- \frac{ E_n+\bar{E}_n}{ T\eta P_n|h_{n,0}|^2},  \frac{T\eta P_n|h_{n,0}|^2-  \bar{E}_n}{T\eta P_n|h_{n,0}|^2 + TP_{\max}}\right\}$,  $\kappa_1=\frac{ \eta P_n|h_{n,0}|^2 |g_{0,n}|^2}{1 + P_n|h_n|^2}$,  $\kappa_2=\frac{ \bar{E}_n|g_{0,n}|^2}{T(1 + P_n|h_n|^2)}  $, and $W_0(\cdot)$ denotes the principle branch of the Lambert W function \cite{GRADSHTEYN}. 
\end{lemma} 
\begin{proof}
See Appendix \ref{proof3}. 
\end{proof}
{\it Remark 6:} Given the closed-form expression for $\alpha_n^*(\bar{E}_n)$, the optimal power allocation coefficient is given by
\begin{align}\label{p0n}
P_{0,n}^*(\bar{E}_n)= \frac{(1- \alpha_n^*) \eta P_n|h_{n,0}|^2}{\alpha_n^*} -  \frac{\bar{E}_n}{\alpha_n^*T},
\end{align}
if $\bar{E}_n\neq T\eta P_n|h_{n,0}|^2$. We note that the case of $\bar{E}_n= T\eta P_n|h_{n,0}|^2$ is not discussed in the proof of Lemma \ref{lemma1}. As discussed in Remark 5, if $\bar{E}_n= T\eta P_n|h_{n,0}|^2$, the entire  time available    will be used for energy harvesting, and there is no data transmission, i.e., $\alpha_n^*(\bar{E}_n)=0$. For this special case, there is no need to specify  the value of $P_{0,n}^*(\bar{E}_n)$.

\section{A DDPG Approach to Optimize $\bar{E}_n$}\label{section V}
\subsection{A Brief Introduction to DDPG}
The ultimate goal of  DDPG is to   determine  an action, denoted by $a$,  which can maximize the action-value function, denoted by $Q(s,a)$, for a given state, denoted by $s$, via the following maximization problem \cite{suttonrl}:
\begin{align}\label{qsa1}
a^*(s) =\arg  \underset{a}\max \quad Q(s,a), 
\end{align}
which is similar to tabular  reinforcement learning algorithms, such as Q-learning and state-action-reward-state-action (SARSA) \cite{9114970}. But unlike Q-learning and SARSA, for DDPG, there is no need to build a table to store the values of  $Q(s,a)$. Instead, the action-value function is approximated by using neural networks, which is similar to    deep Q networks (DQN). Furthermore,  unlike Q-learning, SARSA, and DQN, DDPG is designed for the case when actions are continuous variables.

In particular,  the four neural networks used for DDPG are listed in the following \cite{ddpg}:
\begin{itemize} 
\item An actor  network (also termed a policy network), parameterized by $\omega_{\mu}$,    takes $s$ as its input  and outputs  an action which is denoted by  $\mu\left(s|\omega_{\mu}\right)$.

\item A target actor network, parameterized by $\omega_{\mu_t}$,   outputs $\mu_t\left(s|\omega_{\mu_t}\right)$. 

\item A critic network  (also termed a Q network), parameterized by $\omega_{c}$,  takes $s$ and $a$ as its inputs,  and outputs the corresponding state-value function, denoted by $Q\left(s,a|\omega_{c}\right)$.

\item A target critic network, parameterized by $\omega_{c_t}$,    outputs $Q_t\left(s,a|\omega_{c_t}\right)$.
\end{itemize} 
 The key features of DDPG are described in the following.

\subsubsection{Exploration} From  the user's perspective,  the actor network is the most important component since it provides the desired    solution.  In order to encourage the algorithm to explore the  environment, noise is added to the output of the actor network, which means that the action to be taken for state $s$ is given by \cite{ddpg}
\begin{align}
a(s) = \mu\left(s|\omega_{\mu}\right) +n,
\end{align}
where $n$ denotes the exploration noise. 

\subsubsection{Updating the  networks} While only the actor network is used to generate the needed action, the other three networks  are crucial to make   sure that the actor network is properly trained.  Assume that there exists  a tuple $(s, a, r, s_{\_})$, where $r$ is the reward if action $a$ is taken for the current state $s$ and $s_{\_}$ denotes the next state. Based on this  tuple, the networks are updated as follows \cite{ddpg}: 

\begin{itemize}
\item The actor network is trained by maximizing   the state-value function, as shown in \eqref{qsa1}. By using the parameters of the actor and critic networks, the objective function for the maximization problem can be rewritten as    $J\left(\omega_{\mu}\right)=Q\left(s,\left.a=\mu\left(s|\omega_{\mu}\right)\right |\omega_{c}\right)$. Given the fact that the action space is continuous and also assuming that the state-value function is differentiable, the parameters of the actor network, $\omega_{\mu}$, can be updated by carrying out gradient ascent. Note that gradient search requires the derivative of the objective function with respect to  $\omega_{\mu}$, where the following chain rule can be used:
\begin{align}\label{update actor}
\bigtriangledown_{\omega_{\mu}} J\left(\omega_{\mu}\right) = \bigtriangledown_{a} Q\left(s,\left.a\right |\omega_{c}\right)  \bigtriangledown_{\omega_{\mu}}\mu\left(s|\omega_{\mu}\right).
\end{align}
Therefore,   the output of the actor network can be used   as the input of the critic network, and the parameters of the actor network ($\omega_{\mu}$) are updated by maximizing the output of the critic network and fixing  the parameters of the critic network.

\item The critic network plays a crucial role in updating the actor network, as can be seen from \eqref{update actor}. The update of the critic network relies on   the two target networks. On the one hand,  by using the output of the target actor network as an input of the target critic network,    a target value for the state-value function is obtained as follows:
\begin{align}
y=   r  + \gamma  Q_{t}(s_{\_}, \mu_t\left(s_{\_}|\omega_{\mu_t}\right) | \omega_{c_t}) ,
\end{align}
where $\gamma$ denotes the  discount  parameter. 
On the other hand, another estimate for the state-value function can be generated by using the critic network, i.e., $Q(s, a | \omega_{c})$. The critic network can be updated by minimizing the loss defined as follows:
\begin{align}\label{lx1}
L = \left(y- Q(s, a | \omega_{c})\right)^2.
\end{align}

\item The two target networks have the same structure as their counterparts and their parameters  are   updated   as follows:
\begin{align}\hspace{-1em}
\omega_{c_t} \rightarrow \tau \omega_{c} +(1-\tau) \omega_{c_t},\quad \omega_{\mu_t} \rightarrow \tau \omega_{\mu} +(1-\tau) \omega_{\mu_t},
\end{align}  
where $\tau$ is the soft updating parameter. We note that the two target networks are updated   with a much lower frequency   than   their counterparts. 
\end{itemize}

\subsubsection{Replay Buffer} Similar to DQN, for DDPG, the past experience, i.e., multiple tuples from the past, $(s_i, a_i, r_i, s_{i,\_})$, are stored in a pool, termed the replay buffer. When the networks are updated, a fixed number  of the tuples are randomly selected from the buffer and used for network updating, which means that both \eqref{update actor} and \eqref{lx1} are carried out in a batch mode.

 \subsection{Application of DDPG to CR-NOMA}

 By using the closed-form expression for  $\alpha^*(\bar{E}_n)$ developed in Lemma \ref{lemma1} and $P_{0,n}^*(\bar{E}_n)$ shown in \eqref{p0n}, the long-term  throughput maximization problem can be rewritten as follows:  {
 \begin{problem}\label{pb:4x}
  \begin{alignat}{2}
\underset{\pi_{ \bar{E}_n}}{\rm{max}} &\quad    
 \mathcal{E} \left\{ \left. \sum^{\infty}_{n=1}\gamma^{n-1} \alpha^*(\bar{E}_n) \ln\left( 1+\frac{P_{0,n}^*(\bar{E}_n)|g_{0,n}|^2}{1 + P_n|h_n|^2}
\right)   \right| \pi_{ \bar{E}_n},  {s}_1 \right\}
\label{obj:4x} \\
\rm{s.t.} & \quad  E_{n+1} = \min\left\{E_{\max}, E_n + \bar{E}_n\right\}
\label{1st:4x}
.
  \end{alignat}
\end{problem} }
Because problem \ref{pb:4x} is a function of a single continuous variable, $\bar{E}_n$, it is ideally suited for  the application of DDPG. The key step for the application of reinforcement learning  is to define the state space, the action space, and the reward, as discussed in the following.

 \subsubsection{State space}  The state space   consists of  all possible states. As discussed in Section \ref{section model}, for the considered long-term system throughput maximization problem, a natural choice for the state is $
s_n = \begin{bmatrix}|g_{0,n}|^2& |h_n|^2& |h_{n,0}|^2  & E_n\end{bmatrix}^T$, which   includes the channel gains associated with the primary user served at $t_n$, the secondary user's channel to the base station, and the   energy available at the beginning of $t_n$.

\subsubsection{Action Space} The action space   consists of  all possible actions taken by the secondary user. For the considered throughput maximization problem, it is  natural   to use $\bar{E}_n$ as the action. For example,  if both $|h_n|^2$ and $|h_{n,0}|^2 $ are extremely strong, a good choice for the action is  $\bar{E}_n\geq 0$, such that  data transmission is avoided due to the strong interference caused by $\bar{\rm U}_n$, and to spend more time for energy harvesting due to the strong connection between $\bar{\rm U}_n$ and ${\rm U}_0$.

However, it is important to point out that the range of $\bar{E}_n$ can be quite large. In particular, as shown in the proof of Proposition \ref{proposition1},  the value of $\bar{E}_n$ is bounded by the two following extreme situations:
\begin{align}\label{bound for Enx}
 \underset{\text{No energy harvesting}}{\underbrace{-\min\{E_n,TP_{\max}\}}}\leq \bar{E}_n\leq
  \underset{\text{No data transmission}}{\underbrace{ \min\left\{ E_{\max}-E_n,  T \eta P_n|h_{n,0}|^2 \right\}}},
\end{align}
which can be a very large negative number and a very large positive number, respectively.  
Note that, compared to the existing work in \cite{9108195}, where   only  $-E_n$   is used as the lower bound on $\bar{E}_n$, we also include  $-TP_{\max}$  in the lower bound. This is due to the transmit power constraint $P_{0,n}\leq P_{\max}$, which means that the maximal energy consumed within $T$ seconds is $TP_{\max}$, if   there is sufficient  energy available  at the beginning of $t_n$, i.e., $E_n>TP_{\max}$.

In order to improve the stability of the used neural networks, it is preferable to constrain  the possible choices of the action within a small and fixed range, ideally between $0$ and $1$. By using the upper and lower bounds on $\bar{E}_n$ and introducing an auxiliary variable $\beta_n$, $0\leq \beta_n\leq 1$, $\bar{E}_n$ can be expressed as follows:  
\begin{align}\label{ennew}
\bar{E}_n =& \beta_n \min\left\{ E_{\max}-E_n,   T \eta P_n|h_{n,0}|^2 \right\} - (1-\beta_n)\min\{E_n,TP_{\max}\}. 
\end{align}
It is straightforward to verify that the lower bound in \eqref{bound for Enx} is used if $\beta_n=0$, and the upper bound can be realized by letting $\beta_n=1$. Therefore, $\beta$ is a suitable   action variable for  DDPG networks.  

\subsubsection{Reward} Because the objective function of the   optimization problem in \ref{pb:4x} is based on the secondary user's long-term data rate, it is natural to use $R_n$ as the reward at $t_n$. 
 
 By using the defined state space,   action space, and   reward, DDPG can be applied straightforwardly. In particular, each episode consists of multiple steps/iterations. During each step, the DDPG algorithm first generates an action according to the current state, finds the next state according to the chosen action, and then the DDPG agent starts learning by updating the four neural  networks as discussed in the previous subsection. The details for the DDPG implementation of DDPG can be found in \cite{myddpgcode}. 
 
%
%

 \section{ Simulation Results}\label{section VI}
 In this section, we study the performance of the proposed DDPG assisted NOMA transmission strategy by using computer simulation results. In our simulations, the learning rates for the actor and critic networks are set as $0.002$ and $0.004$, respectively. The reward discount parameter is $\gamma=0.9$, the network updating parameter is $\tau=0.01$, and the batch size for the replay experience is $32$.    The base station is located at the origin of the x-y plane, and the secondary user is located at ($1$ m, $1$ m).  The primary users' transmit power is set as $P_n=30$ dBm, the initial energy in the battery is $E_{\max}=0.1$ J, $T=1$ s, $P_{\max}=0.1$ W, and the energy harvesting efficiency coefficient is set as $\eta=0.7$. The additive white Gaussian noise power spectral density is $-170$ dBm/Hz,   the total bandwidth is $1$ MHz,   the carrier frequency is $914$ MHz, the path loss model in \cite{127405} is used, and  the path loss exponent is $3$. 
 
Regarding   the neural networks, a simple neural network which consists of  $2$ hidden layers with $N_n=64$ nodes  each is used for the actor network, where  the     rectified linear activation function (ReLU) is used in the first hidden layer, and the hyperbolic tangent function is used for the second hidden layer and the output layer.  Since the critic network takes two inputs,  $s$ and $a$,   both   $s$ and $a$ are fed to two individual hidden layers with $64$ nodes each before they are concatenated and connected to another hidden layer with $64$ nodes.  ReLU is used in all the  layers of the critic network. The two target networks are built in the same manner as their counterparts.   The detailed  setups for the four employed neural networks can be found in \cite{myddpgcode}.  {We note that the use of DDPG results in a higher  computational complexity compared to  the two benchmark schemes outlined in Remark 2. Take the greedy  scheme as an example, which uses   all   available energy  for data transmission, and hence is computationally efficient to implement. For DDPG,   if the used simple actor network is fully trained, generating $\bar{E}_n$ given $s_n$ requires a moderate  computational complexity  of  $ \mathcal{O} \left( N_n^2\right)$, since the numbers of inputs,  outputs, and layers are much smaller than $N_n$, where $\mathcal{O}\left(\cdot\right)$ denotes the computational complexity operator  \cite{deeplearningmit}. The complexity analysis of the training stage  is challenging not only because it   depends on how many gradient iterations are carried out but also because the employed  critic network is not fully connected. Due to space limitations, a full complexity analysis and methods to reduce complexity will be treated as an important direction for future research.  It is expected that  the training stage of the DDPG entails a  high computational complexity, where a promising energy and computationally efficient mitigation  approach is to apply mobile edge computing (MEC) and allow  the secondary user to offload computations to the base station \cite{8016573}.     }

The concepts of steps, episodes, and experiments for  DDPG can be interpreted in the considered CR-NOMA context as follows. Each step represents a time slot,  each episode consists of $100$ time slots (or steps), and each experiment consists of multiple episodes.  At the beginning of each experiment, the neural networks are randomly initialized.  At the beginning of each episode, the secondary user's battery is reset to $E_{\max}$. 
   \begin{figure}[!t]\vspace{-3em}
\begin{center} \subfigure[  Experiment  I ]{\label{fig 0 a}\includegraphics[width=0.49\textwidth]{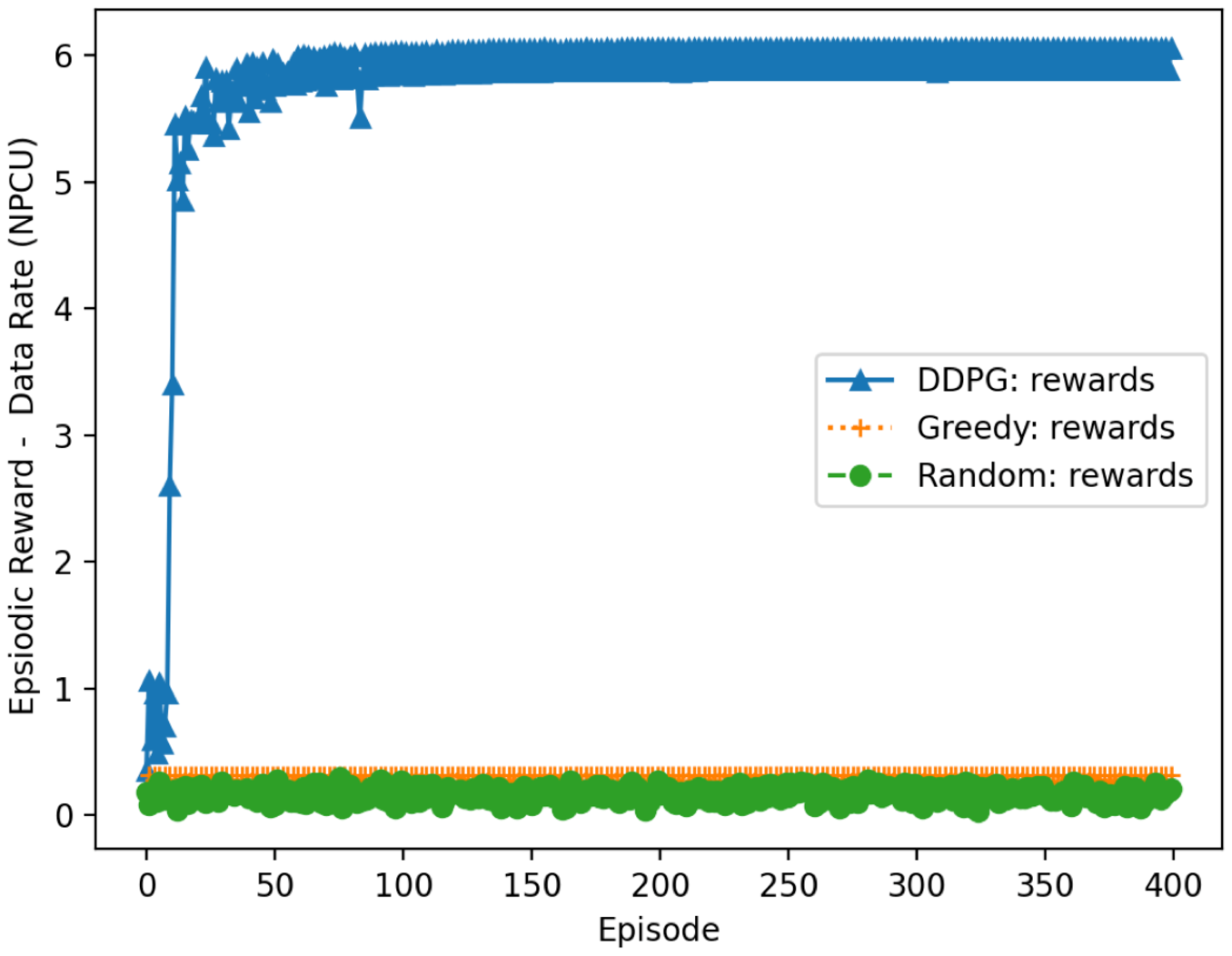}}
\subfigure[ Experiment  II]{\label{fig 0 b}\includegraphics[width=0.49\textwidth]{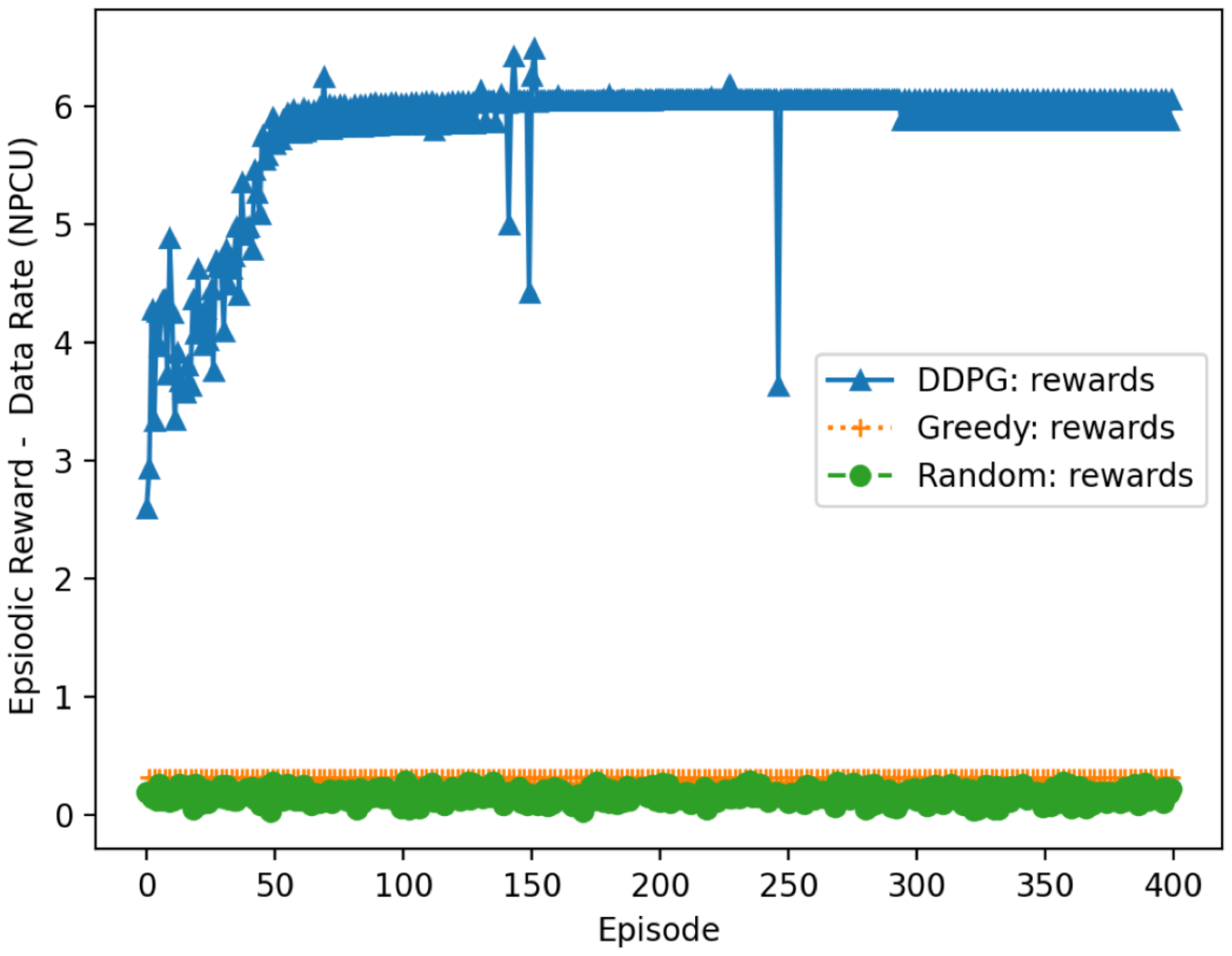}} 
\end{center} \vspace{-1em}
\caption{ A deterministic extreme case with $M=2$ primary users. The two primary users are located at ($0$ m, $1$ m)  and ($0$ m, $1000$ m).   Only large-scale path loss is considered and random fading is omitted.  The difference between the two experiments  is caused by the fact that the neural networks are randomly initialized   at the beginning  of each experiment.     \vspace{-1em} }\label{fig4}\vspace{-1em}
\end{figure}
 \subsection{A  Two-User Deterministic     Case }
 In order to gain insight into the performance of DDPG, we first focus on a deterministic case with two primary users, where  the two primary users are located at ($0$ m, $1$ m)  and ($0$ m, $1000$ m),   only large-scale path loss is considered, and random fading is omitted. The rationale behind this setting is the study of   an extreme scenario, in which ${\rm U_1}$ has a strong channel  to both the base station and ${\rm U_0}$, but ${\rm U_2}$'s channels  to both nodes are extremely weak. An intuitive  decision for this case is to ask ${\rm U_0}$ to carry out energy harvesting only, when ${\rm U_1}$ transmits, which has the following two benefits: severe interference from ${\rm U_1}$ is avoided and   a large amount of energy from ${\rm U_1}$ is harvested. When ${\rm U_2}$ transmits, an intuitive  decision  is to avoid energy harvesting at ${\rm U_0}$ since only a small amount of energy can be harvested due to the severe path loss between ${\rm U_0}$ and ${\rm U_2}$. Instead, ${\rm U_0}$ should carry out data transmission only, since ${\rm U_2}$ will not cause much interference. With this intuitive  decision, the amount of energy harvested when ${\rm U_1}$ transmits is given by \cite{127405}
 \begin{align}\label{enx1}
 E=&(1- \alpha_n) T \eta P_n|h_{n,0}|^2=   0.7  /10^{3.17}  \text{ J}= 7*10^{-4.17} \text{ J}\approx 4.7*10^{-4} \text{ J}.
 \end{align}
When ${\rm U_2}$ transmits, the entire   time slot   is used for data transmission    using the harvested energy shown in \eqref{enx1}. Therefore, the available transmit power is $\frac{E}{T}=4.7*10^{-4} $ W,  which means that   the following average   data rate, denoted by $\bar{R}_n$, is achievable by the secondary user:
\begin{align} 
\bar{R}_n =&\frac{1}{2} \alpha_n \ln\left( 1+\frac{P_{0,n}|g_{0,n}|^2}{1 + P_n|h_n|^2}
\right)\\ \label{ratencc}
 =& \frac{1}{2}  \ln\left( 1+\frac{ 4.7*10^{-4}\times \frac{1}{10^{3.17}*2^{\frac{3}{2}}}}{10^{-14} + 10^{-3.17}\times 1000^{-3}}
\right) \approx 6.01 \text{ NPCU},  
\end{align}
where  the factor $\frac{1}{2}$ is used since ${\rm U_0}$ transmits only when  ${\rm U_2}$ transmits.  Comparing  the result in \eqref{ratencc} with Fig. \ref{fig4}, one can   observe that the   proposed reinforcement learning scheme   realizes this intuitive  decision. We note that DDPG might perform slightly differently for different   random initializations, as shown in Figs. \ref{fig 0 a} and \ref{fig 0 b}. On the other hand,     the benchmark schemes yield much worse  performances  than  the proposed scheme, since the decisions they make are not based on a long-term objective. For example, when   ${\rm U_1}$ transmits, the greedy algorithm will still try to facilitate  data transmission, which results in an inefficient  energy use due    to the strong interference caused by ${\rm U_1}$.

  \begin{figure}[!t]\vspace{-3em}
\begin{center} \subfigure[ Experiment I ]{\label{fig 3 a}\includegraphics[width=0.49\textwidth]{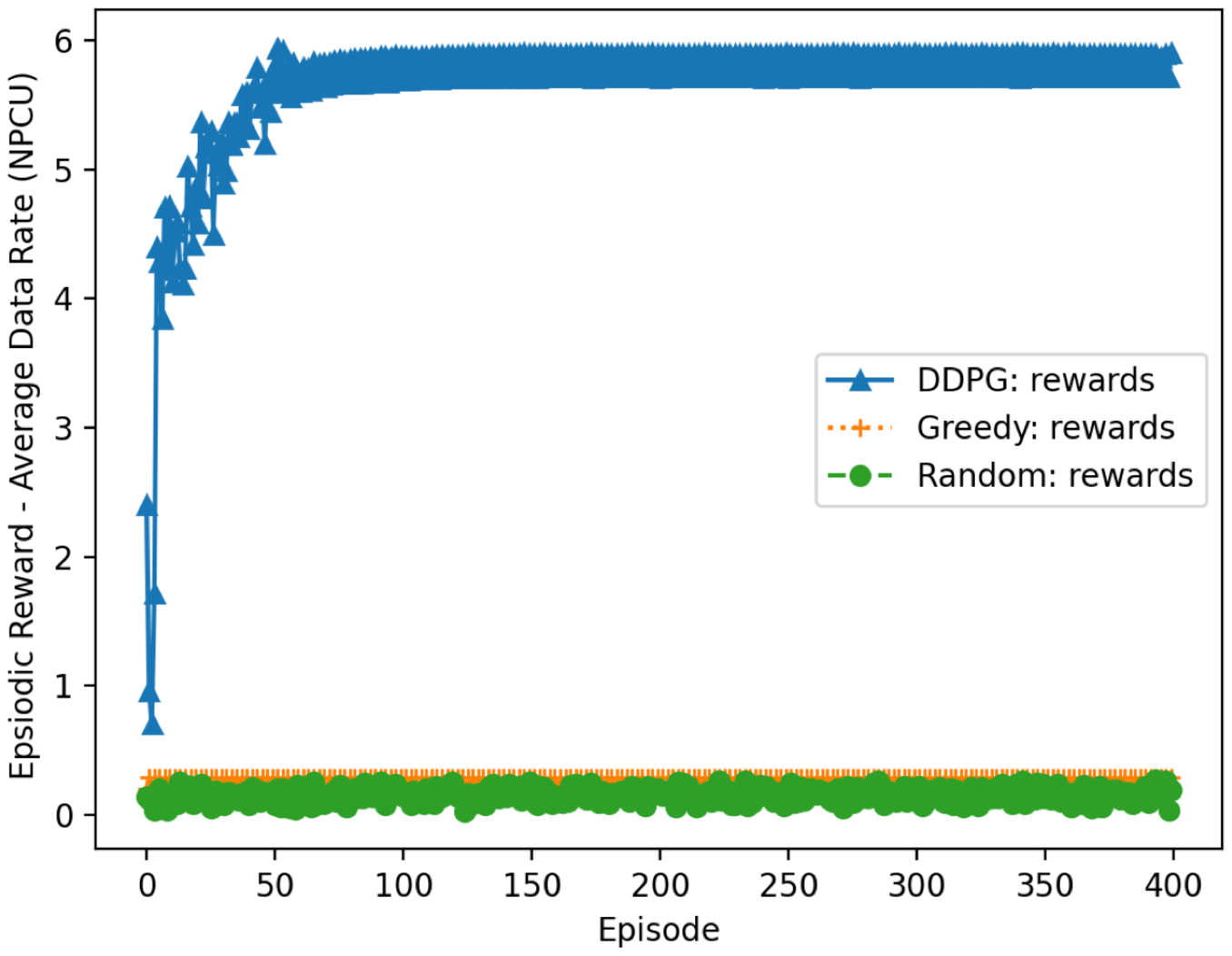}}
\subfigure[Experiment II  ]{\label{fig 3 b}\includegraphics[width=0.49\textwidth]{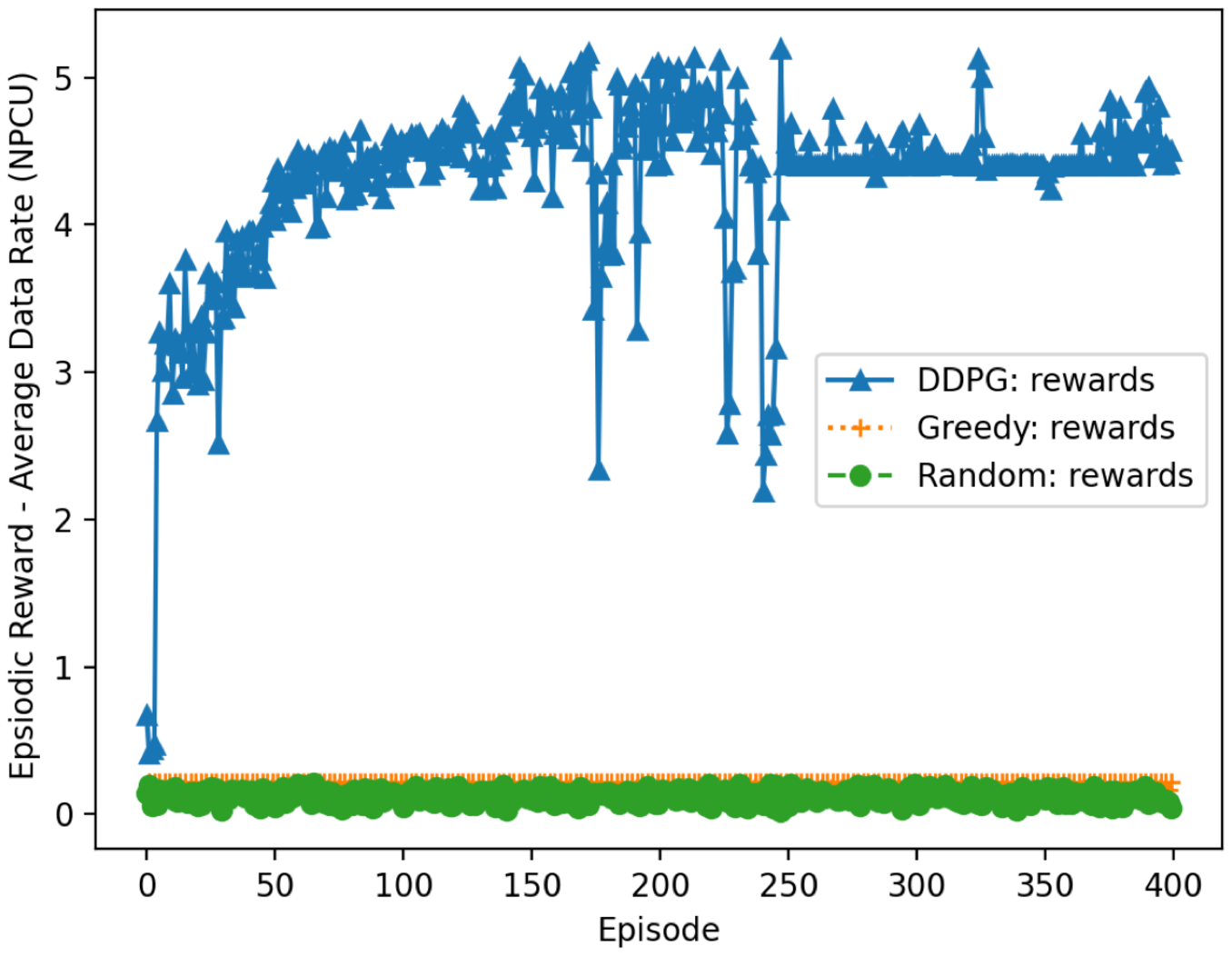}} 
\end{center}\vspace{-1em}
\caption{ The performance of the proposed DDPG scheme with $M=2$ primary users.    Both large-scale path loss and small-scale multi-path fading are considered. The two experiments are based on  different small-scale fading realizations.  }\label{fig3x}\vspace{-2em}
\end{figure}

  \begin{figure}[!b]\vspace{-2em}
\begin{center} \subfigure[ Experiment I  ]{\label{fig 6x a}\includegraphics[width=0.49\textwidth]{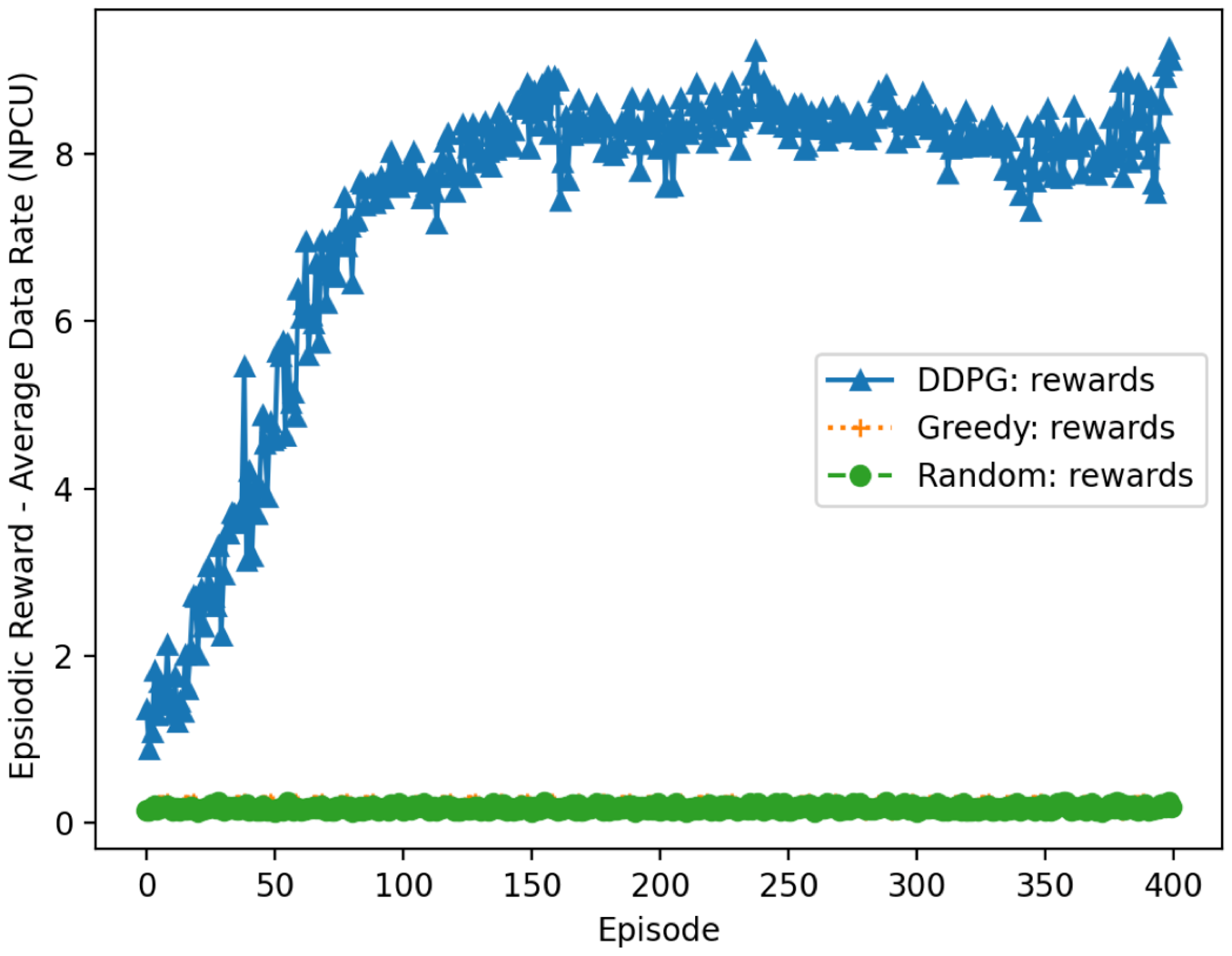}}
\subfigure[  Experiment II  ]{\label{fig 6x b}\includegraphics[width=0.49\textwidth]{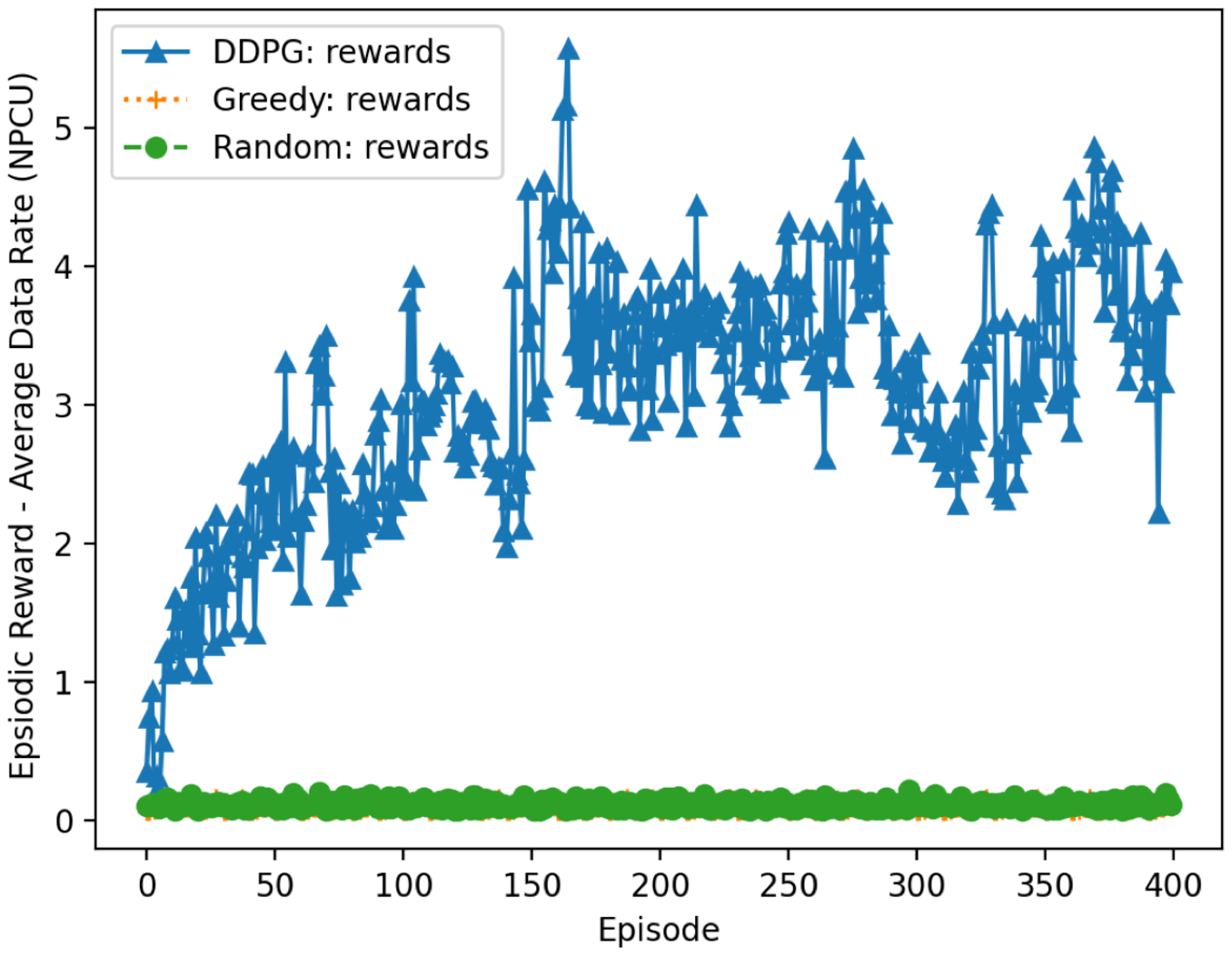}}\vspace{-1em}
\end{center}
\caption{    The performance of the proposed DDPG scheme with $M=10$ primary users.    Both large-scale path loss and small-scale multi-path fading are considered. The two experiments are based on  different small-scale fading realizations.  }\label{fig6x}\vspace{-2em}
\end{figure}
 \subsection{  General Multi-User Cases with  Constant Fading}
 In this subsection, the performance of the proposed reinforcement learning scheme is studied for a more general case with multiple users. In particular, the $M$ users are equally spaced  on a segment between ($1$ m, $0$ m) and ($1000$ m, $0$ m).  {Furthermore, small-scale multi-path fading is also considered, in addition to path loss. In order to study the convergence of the proposed DDPG algorithm,  during each experiment which contains multiple episodes,  the users' channels are kept constant, where independent and identically distributed (i.i.d.) complex Gaussian random variables with zero means and unit variance are  used to simulate the small-scale channel fading.} 
 
  In Figs. \ref{fig3x} and \ref{fig6x},  $M=2$  and $M=10$ primary users are considered, respectively. On the one hand, comparing    Fig. \ref{fig3x} to Fig. \ref{fig4}, a data rate reduction can be observed, which is due to the consideration of channel fading. On the other hand, comparing      \ref{fig 6x a} to  Fig. \ref{fig3x}, it is interesting to observe that the case with   $M=10$ can yield a larger data rate than the case with $M=2$. This is due to the fact that for $M>2$, it is possible for the secondary user to use more than half of the time for data transmission.    As can be observed by comparing the respective  subfigures of Fig. \ref{fig3x}  and \ref{fig6x}, different data rates are realized for a given $M$, which is due to the fact that the   subfigures employ  different realizations of   the small-scale  fading. However, regardless of the realizations of the random fading, the proposed DDPG scheme can always achieve  a considerably larger data rate than the two benchmark schemes, as is evident from the figures.

%

  \begin{figure}[!bt] \vspace{-2em}
\begin{center} \subfigure[  $M=2$ ]{\label{fig 0 a5}\includegraphics[width=0.49\textwidth]{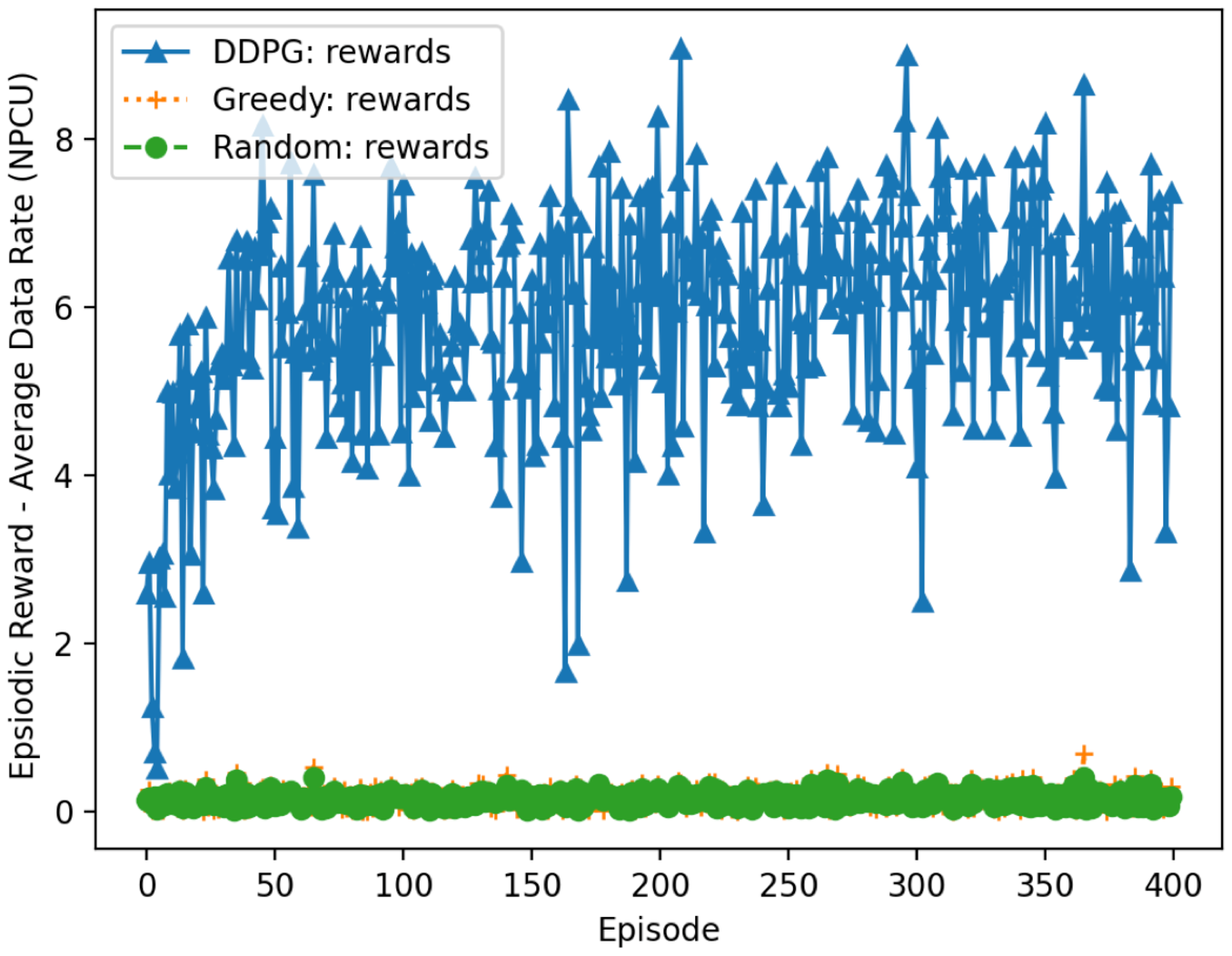}}
\subfigure[ $M=10$]{\label{fig 0 b5}\includegraphics[width=0.49\textwidth]{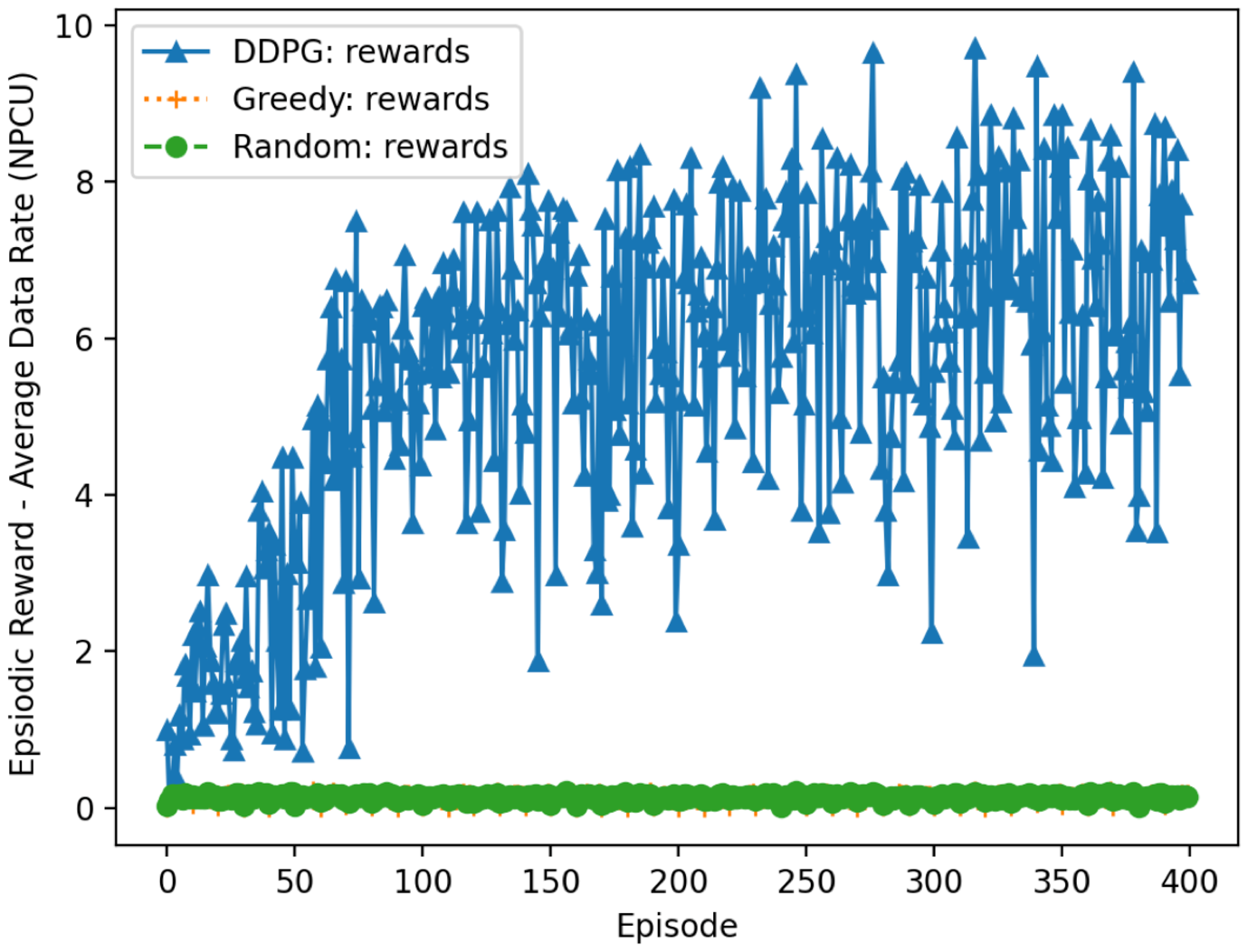}} 
\end{center} \vspace{-1em}
\caption{ The performance of the proposed DDPG scheme in time-varying scenarios, where     different small-scale fading realizations are used in  different episodes. }\label{fignew1} \vspace{-1em}
\end{figure}
Recall that  for the deterministic case with $M=2$ considered in Fig. \ref{fig4}, the decision generated by the proposed DDPG scheme converges quickly, e.g.,   after $50$ episodes the DDPG curves shown in Fig. \ref{fig4} become   almost flat. This property  disappears in general  for the cases with random fading.   In particular, depending on the random fading realization,  the DDPG agent might lead to a constant decision policy  after  a sufficiently large  number of episodes, as shown in Figs. \ref{fig 3 a} and \ref{fig 6x a}. However, Figs. \ref{fig 3 b} and \ref{fig 6x b} also show that for certain  random fading realizations, the DDPG algorithm may converge very slowly. One possible reason for this observation is that for the  two-user deterministic case considered in Fig. \ref{fig4}, the underlying  pattern is simple, e.g., one user has   strong channels  but the other user does not. This simple  pattern can be quickly and efficiently  learned by  DDPG. When there are more users packed in the same area and/or the   fading is random, it is possible that such a clear pattern does not exist, i.e.,   multiple users' channel conditions are similar to each other. As a result, there is no clear strategy for energy harvesting and data transmission.

 \subsection{  General Multi-User Cases with  Time-Varying Fading}
 {Finally, the performance of the proposed DDPG scheme is studied in a general multi-user scenario  with time-varying channels.   For the   simulation results presented in the following, the locations of the users are fixed in the same manner as in the previous subsection, and different small-scale fading realizations are used in  different episodes.    In Fig. \ref{fignew1},  the performance of the proposed DDPG algorithm is shown as a function of the number of episodes in the considered time-varying scenario. As can be observed from the figure, the use of the proposed DDPG algorithm can still offer a significant performance gain over the two benchmark schemes, particularly if  a sufficient number of episodes is employed. In addition, compared to the case $M=2$, the case $M=10$ requires more   episodes to  ensure a significant  performance gain over the benchmark schemes.   These observations  are consistent with those made for time-invariant channels. We   note that  the episode reward curves in  Fig. \ref{fignew1} show more variation   than those for the time-invariant case. This is mainly due to the fact that different fading realizations are  experienced  in  different episodes. As a consequence,   the maximal throughputs in different episodes are also expected to be different, which means that the episode reward curves cannot be flat, even if the proposed algorithm can achieve the respective maximal throughputs.      }

  \begin{figure}[t]\vspace{-3em}
\begin{center} \subfigure[  The cases with $1$,  $10$, and $20$ episodes]{\label{fig 7x a}\includegraphics[width=0.44\textwidth]{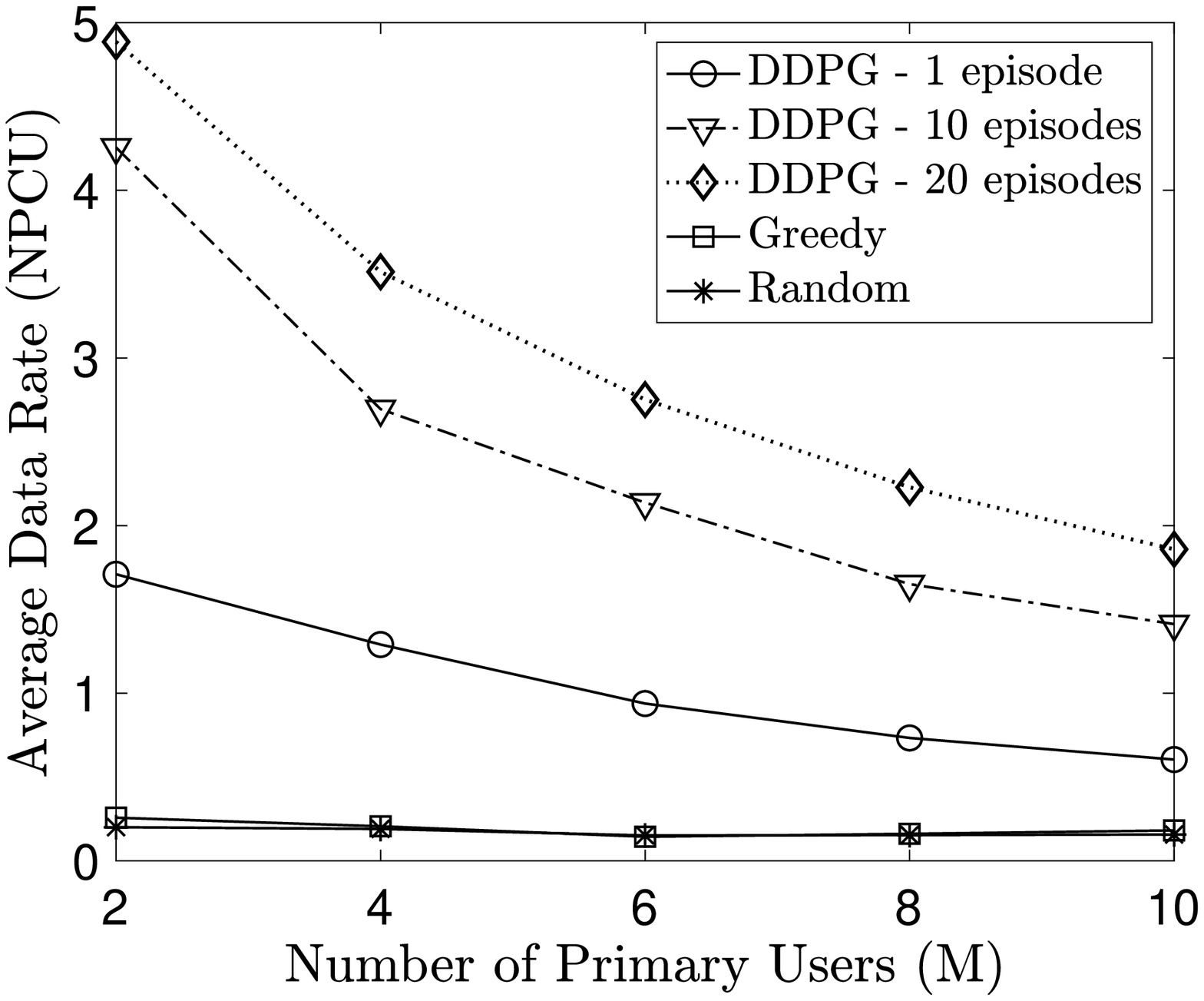}}
\subfigure[   The cases with $50$, $100$, and $150$ episodes]{\label{fig 7x b}\includegraphics[width=0.44\textwidth]{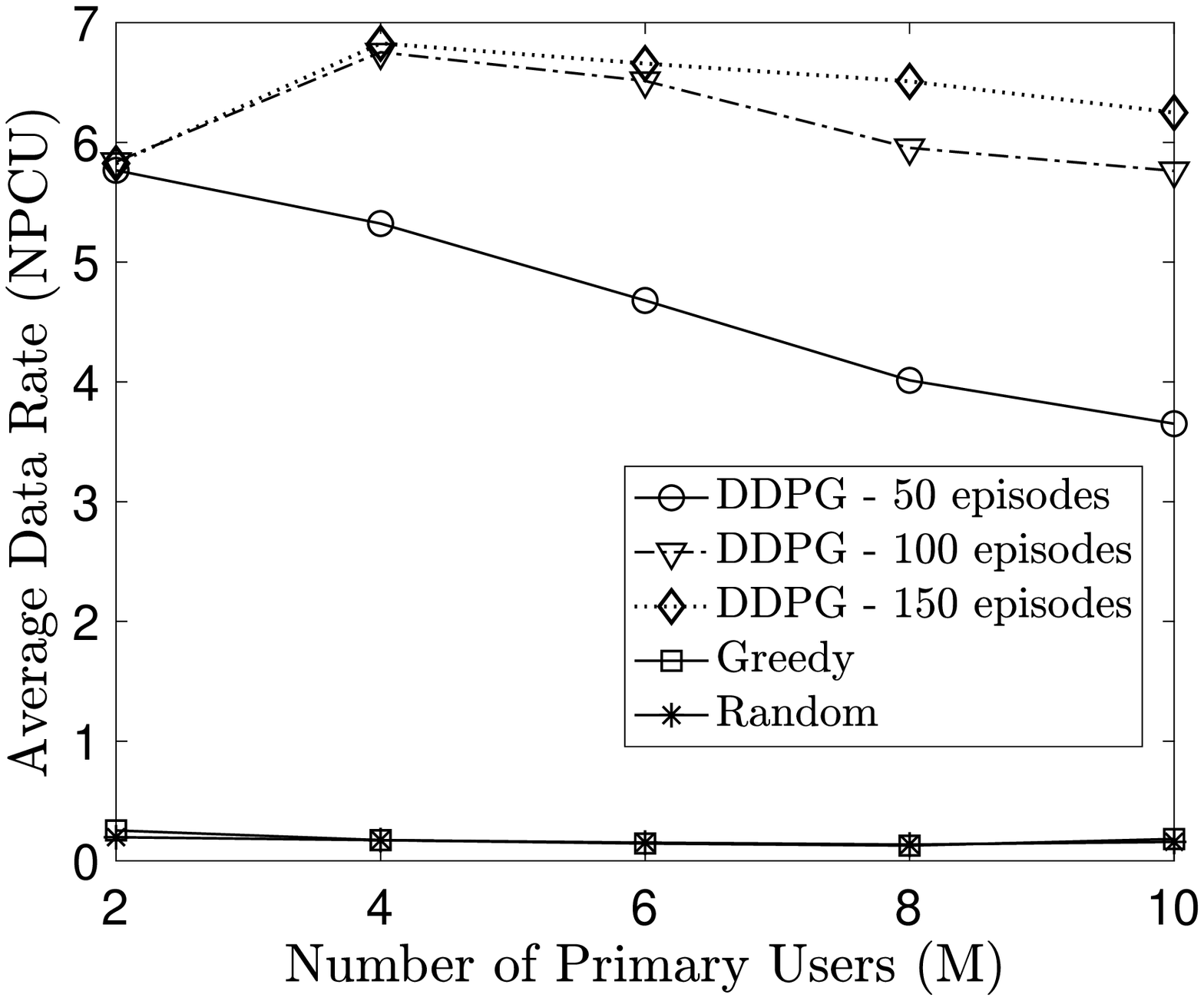}}\vspace{-1em}
\end{center}
\caption{    Impact of $M$ and the number of the used episodes  on the performance of the proposed DDPG scheme in time-varying scenarios, where     different small-scale fading realizations are used in  different episodes.     }\label{fig7}\vspace{-2em}
\end{figure}

 {In Fig. \ref{fig7}, the performances of the three considered    schemes are compared  by using the average data rate as the metric, where    Monte Carlo simulations are carried out to average out the randomness caused by time-varying fading. In Fig. \ref{fig 7x a}, three different choices for the number of    episodes are used, in order to illustrate  how the performance of the DDPG scheme is affected by  the amount of training. As can be observed from Fig. \ref{fig 7x a}, even for a single episode, the  proposed DDPG algorithm can already realize a larger data rate than the greedy and random schemes. By increasing the number of episodes from $1$ to $10$ and $20$, the performance gain of the proposed DDPG algorithm over the benchmark schemes can be further  improved, as shown in Fig. \ref{fig 7x a}.   Recall that we observed   from Figs. \ref{fig3x} and \ref{fig6x} that  the proposed DDPG algorithm starts to converge after $100$ episodes. This is the reason why  increasing the number of episodes in  Fig. \ref{fig 7x b}  from $100$ to $150$ improves the performance of the proposed DDPG scheme  only slightly.   Another important observation from   Fig. \ref{fig7} is that the case with more primary users participating in the NOMA transmission requires a larger amount of training. This observation is expected since a case with more primary users is more complex  and hence more episodes are required to train the DDPG algorithm.  
}

\section{Conclusions}
In this paper,    machine learning has been applied  in CR-NOMA networks to facilitate spectrum and energy cooperation between  multiple primary users and an energy-constrained secondary user. The goal of the machine learning algorithm  is to   maximize the secondary user's long-term throughput, which  is challenging   due to the need for making decisions which yield long-term gains but might result in short-term losses.  For example, when, in a given time slot,  a primary user
with large channel gains transmits, an intuitive approach  is that the secondary user should not carry out data transmission but perform energy harvesting only, which results in  zero data rate for this time slot but yields benefits in the  long-term. Here, the considered long-term throughput maximization problem was first reformulated, and then  a DRL  approach, termed DDPG, was applied to emulate the intuitive solution. The provided simulation results demonstrate that the proposed deep reinforcement learning assisted NOMA transmission scheme can yield  significant performance gains over two benchmark schemes. 

 {While the case with  a single  secondary user has been the focus of  this paper,  the proposed idea of   spectrum and energy sharing among primary and secondary users can be extended to   scenarios with multiple secondary users.   For example, the scheme in \cite{bacnomamtc}  shows the possibility of carrying  out spectrum and energy cooperation among one primary downlink user and multiple secondary uplink users. Compared to the case with a single secondary user, the  resource allocation design  for multi-user systems is more challenging, due to the constraint that additional secondary users should be served  without affecting the primary users' QoS experience. In addition, the interference between the secondary users has to be carefully controlled.  Furthermore, in this paper,  a linear energy harvesting model has been assumed,  and the impact of the energy consumed  for signal processing and by the radio frequency circuits has not been taken into consideration. These assumptions have allowed us to gain important    insights regarding 	the long-term optimization of  the two conflicting tasks, energy harvesting and data transmission.   However,  the extension of  the proposed machine learning algorithm to the case with   more practical assumptions      is an important direction for future research. Moreover,  it was assumed in this paper that the energy constrained secondary user carries out radio-frequency energy harvesting. Alternatively, the secondary user can also use backscatter communications (BackCom) to realize battery-less transmission, and the application of DRL to BackCom assisted NOMA transmission is another interesting  direction for future research \cite{wongwcnc20}.   } 
\appendices
\section{Proof for Proposition \ref{proposition1}}\label{proof1}

In order to analyze the feasibility of problem \ref{pb:7}, the range of  $\bar{E}_n$ needs to be studied first.  Recall that $\bar{E}_n$ is the difference between the energy harvested and consumed at $t_n$, i.e., $(1-\alpha_n) T \eta P_n|h_{n,0}|^2 - \alpha_n TP_{0,n}$. The smallest value achievable by $\bar{E}_n$ is caused by the case with no energy harvesting, i.e., all the time is  used for data transmission and $\alpha_n=1$, which means that 
\begin{align}\label{bound for En1xx}
\bar{E}_n\geq  -   TP_{0,n}.
\end{align}
Furthermore,   the energy available at the beginning of $t_n$ is $E_n$, which means that the energy available for data transmission is capped by $E_n$, which results in another lower bound on  $\bar{E}_n$, i.e.,    $\bar{E}_n\geq -E_n$. As a result, $\bar{E}_n$ can be lower bounded as follows:
\begin{align}\label{bound for En1x}
\bar{E}_n\geq-\min\{E_n,TP_{\max}\},
\end{align} 
where we replaced  $P_{0,n}$ by $P_{\max}$ in \eqref{bound for En1xx}.
On the other hand, the largest value achievable by $\bar{E}_n$ is caused by the decision not to carry out any data transmission but perform energy harvesting only, i.e., $\alpha_n=0$. In this case, an upper bound on $\bar{E}_n$ can be obtained as follows:
\begin{align}
\bar{E}_n\leq    T \eta P_n|h_{n,0}|^2.
\end{align}
Moreover,  there is an upper bound on the amount of energy stored  at $t_n$ because of the finite battery capacity. In particular, given the capacity of the battery, $E_{\max}$, and the energy available in the battery at the beginning of $t_n$, $E_n$, the maximal amount of new energy that can be stored  at $t_n$ is capped by $E_{\max}-E_n$. As a result,   $\bar{E}_n$ is upper bounded as follows: 
\begin{align}\label{bound for En2x}
 \bar{E}_n\leq \min\left\{ E_{\max}-E_n,  T \eta P_n|h_{n,0}|^2 \right\}.
\end{align}
By combining \eqref{bound for En1x} and \eqref{bound for En2x},  $\bar{E}_n$ can be bounded as follows: 
\begin{align}\label{bound for En}
 -\min\{E_n,TP_{\max}\}\leq \bar{E}_n\leq \min\left\{ E_{\max}-E_n,  T \eta P_n|h_{n,0}|^2 \right\}.
\end{align}

The feasibility of problem \ref{pb:7} is illustrated in the following two steps. The first step is to show that the lower bounds on $\alpha_n$ for problem \ref{pb:7} do not conflict with the upper bounds. The second step is to show that the intersection of the set defined by the constraints of problem \ref{pb:7} and the domain of its objective function is  non-empty.

Note that there are three lower bounds defined by \eqref{2st:7}, \eqref{3st:7}, and \eqref{4st:7}, respectively. 
Firstly, we focus on the lower bound $\alpha_n\geq 0$ in \eqref{4st:7}. Apparently, it does not conflict with the upper bound in \eqref{4st:7}, i.e., $\alpha_n\leq  1$. In order to compare this lower bound to the upper bound in \eqref{1st:7}, $   \alpha_n \leq 1-  \frac{\bar{E}_n}{ T\eta P_n|h_{n,0}|^2} $,  we first rewrite  the upper   bound in \eqref{bound for En} as follows:
\begin{align}
 \bar{E}_n\leq \min\left\{ E_{\max}-E_n,  T \eta P_n|h_{n,0}|^2 \right\}\leq    T \eta P_n|h_{n,0}|^2.
\end{align} 
Therefore, $1-  \frac{\bar{E}_n}{ T\eta P_n|h_{n,0}|^2} \geq 0$, and hence one can conclude that the lower bound  $\alpha_n\geq 0$ does not conflict with the upper bound in \eqref{1st:7}. 

Secondly, we focus on the lower bound in \eqref{2st:7}, $ \alpha_n   \geq 1- \frac{ E_n+\bar{E}_n}{ T\eta P_n|h_{n,0}|^2}$.  To prove  that it does not conflict with the upper bound in \eqref{4st:7}, $\alpha_n\leq 1$, it is sufficient to show the following
\begin{align}
1- \frac{ E_n+\bar{E}_n}{ T\eta P_n|h_{n,0}|^2}\leq 1 ,
\end{align}
which requires $  \bar{E}_n \geq -E_n$. This   always holds because the use of   the   lower bound in \eqref{bound for En} yields    $  \bar{E}_n\geq -\min\{E_n,TP_{\max}\} $ or equivalently $  \bar{E}_n\geq \max\{-E_n,-TP_{\max}\} \geq -E_n$. In addition, it is straightforward to show that the lower bound in \eqref{2st:7} is always smaller than the upper bound in \eqref{1st:7} by using the fact that $E_n\geq 0$. 
 
Thirdly, we focus on the lower bound in \eqref{3st:7}, $\alpha_n\geq \frac{T\eta P_n|h_{n,0}|^2-  \bar{E}_n}{T\eta P_n|h_{n,0}|^2 + TP_{\max}}$.  To prove  that it does not conflict with the upper bound in \eqref{4st:7}, it is sufficient to show the following
\begin{align}
\frac{T\eta P_n|h_{n,0}|^2-  \bar{E}_n}{T\eta P_n|h_{n,0}|^2 + TP_{\max}} \leq 1 ,
\end{align}
which requires $  \bar{E}_n \geq -   TP_{\max}$. This also always holds because the use of   the   bound in \eqref{bound for En} yields   $  \bar{E}_n\geq \max\{-E_n,-TP_{\max}\} \geq -TP_{\max}$. In order to  show that the lower bound in \eqref{3st:7} is always smaller than the upper bound in \eqref{1st:7}, we need to prove the following inequality
\begin{align}
\frac{T\eta P_n|h_{n,0}|^2-  \bar{E}_n}{T\eta P_n|h_{n,0}|^2 + TP_{\max}} \leq 1-  \frac{\bar{E}_n}{ T\eta P_n|h_{n,0}|^2} ,
\end{align}
which can be simplified as follows:
\begin{align}
\bar{E}_n \leq    T\eta P_n|h_{n,0}|^2.    
\end{align}
This inequality always holds since $\bar{E}_n\leq \min\left\{ E_{\max}-E_n,  { T \eta P_n|h_{n,0}|^2 }\right\}\leq T \eta P_n|h_{n,0}|^2 $ based on   \eqref{bound for En}. 
In summary,   the constraints of problem \eqref{pb:7} do not conflict with each other, and the set defined by these constraints is not empty. 

The last step to analyse the feasibility of problem \ref{pb:7} is to show that the intersection of the domain of the objective function of problem \ref{pb:7} and the set defined by  the constraints is not empty.  Recall that the domain of the objective function \eqref{obj:7} imposes the following constraint on $\alpha_n$:
\begin{align}\label{uob7}
 1+  \frac{(1- \alpha_n) }{\alpha_n}\frac{ \eta P_n|h_{n,0}|^2|g_{0,n}|^2}{1 + P_n|h_n|^2}-  \frac{1 }{\alpha_n} \frac{\bar{E}_n|g_{0,n}|^2}{T(1 + P_n|h_n|^2)}\geq 0.
\end{align}
As can be shown, the constraint imposed on $\alpha_n$ by   problem \ref{pb:7} is stricter than \eqref{uob7}. In particular, one can convert the following inequality 
\begin{align}\label{uob7x}
  \frac{(1- \alpha_n) }{\alpha_n}\frac{ \eta P_n|h_{n,0}|^2|g_{0,n}|^2}{1 + P_n|h_n|^2}-  \frac{1 }{\alpha_n} \frac{\bar{E}_n|g_{0,n}|^2}{T(1 + P_n|h_n|^2)}\geq 0,
\end{align}
into the following equivalent form:
\begin{align}\label{uob7y}
 (1- \alpha_n)   T\eta P_n|h_{n,0}|^2  -    \bar{E}_n \geq 0.
\end{align} 
With some algebraic manipulations, one can find that the inequality constraint in \eqref{uob7y} is equivalent to constraint \eqref{1st:7}. In other words, the intersection of the set defined by the constraints of problem \ref{pb:7} and the domain of its objective function is  non-empty.    As a result,     problem \eqref{pb:7} is always feasible, and the proposition is proved.

\section{Proof for Proposition \ref{proposition2}}\label{proof2}

Recall that the objective function of problem \ref{pb:7} can be expressed as follows:
\begin{align}
  \tilde{f}^*_0\left( \alpha_n\right) = & \alpha_n \ln\left( 1+  \frac{(1- \alpha_n) }{\alpha_n}\frac{ \eta P_n|h_{n,0}|^2|g_{0,n}|^2}{1 + P_n|h_n|^2}  -  \frac{1 }{\alpha_n} \frac{\bar{E}_n|g_{0,n}|^2}{T(1 + P_n|h_n|^2)}
\right).
\end{align}
By using \eqref{uob7y} and the fact that $ T\eta P_n|h_{n,0}|^2  \geq    \bar{E}_n $, it can be shown   that $\alpha_n=0$ is in the domain of  $ \tilde{f}^*_0\left( \alpha_n\right)$. 
 
To simplify notations, we define    $\kappa_1=\frac{ \eta P_n|h_{n,0}|^2 |g_{0,n}|^2}{1 + P_n|h_n|^2}$,  $\kappa_2=\frac{ \bar{E}_n|g_{0,n}|^2}{T(1 + P_n|h_n|^2)}  $, and $x=\alpha_n$. Therefore, the objective function   $\tilde{f}^*_0(x)$ can be expressed as follows:
\begin{align} \nonumber
\tilde{f}^*_0(x) &=   x \ln\left( 1+  \frac{1-x}{x} \kappa_1 - \frac{1}{x}\kappa_2
\right)\\\label{sigularity} &=
 x \ln\left( 1-\kappa_1+  \frac{ \kappa_1-\kappa_2}{x} 
\right).
\end{align}
As discussed in the proof of Proposition \ref{proposition1},   the term  $1+  \frac{1-x}{x} \kappa_1 - \frac{1}{x}\kappa_2$ is strictly positive once all the constraints in problem \ref{pb:7} are satisfied.  Furthermore, as discussed in Remark 5, it is assumed that   $\bar{E}_n\neq T\eta P_n|h_{n,0}|^2$. Therefore, the first order derivative of $\tilde{f}^*_0(x) $ is given by
\begin{align}\label{first orderxx}
\frac{d\tilde{f}_0^*(x)}{dx} = &   \ln\left( 1-\kappa_1+  \frac{ \kappa_1-\kappa_2}{x} \right)- \frac{\frac{\kappa_1-\kappa_2}{x}}{ 1-\kappa_1+  \frac{ \kappa_1-\kappa_2}{x} }\\ \nonumber
= &   \ln\left( 1-\kappa_1+  \frac{ \kappa_1-\kappa_2}{x} \right)- \frac{(\kappa_1-\kappa_2)}{ x(1-\kappa_1)+    \kappa_1-\kappa_2  } .
\end{align}
The  second order derivative of $\tilde{f}^*_0(x) $ is given by
\begin{align}
\frac{d^2\tilde{f}^*_0(x)}{dx^2} = &   - \frac{(\kappa_1-\kappa_2)}{ x(x(1-\kappa_1)+    \kappa_1-\kappa_2 ) }  + \frac{(1-\kappa_1)(\kappa_1-\kappa_2)}{ (x(1-\kappa_1)+    \kappa_1-\kappa_2 )^2 }  .
\end{align}
With some algebraic manipulations, the second order derivative can be expressed as follows:
\begin{align}
\frac{d^2\tilde{f}^*_0(x)}{dx^2} = &     \frac{-(\kappa_1-\kappa_2)^2}{ (x(1-\kappa_1)+    \kappa_1-\kappa_2 )^2 x}  \leq  0 ,
\end{align}
for $x\geq 0$. Therefore, the objective function is a concave function for $\alpha_n\geq 0$ and the proof is complete.  

\section{Proof for Lemma \ref{lemma1}}\label{proof3}

 Problem \ref{pb:7} can be expressed in the following compact form 
\begin{problem}\label{pb:8}
  \begin{alignat}{2}
\underset{  \alpha_n}{\rm{max}} &\quad  \alpha_n \ln\left( 1+  \frac{(1- \alpha_n) }{\alpha_n}\kappa_1-  \frac{1 }{\alpha_n} \kappa_2
\right)
\label{obj:8} \\
\rm{s.t.} &  \quad 
   \alpha_n \leq \min\left\{1, 1-  \frac{\bar{E}_n}{ T\eta P_n|h_{n,0}|^2}\right\}  \label{1st:8} \\
  &  \quad 
  \alpha_n   \geq \max\left\{0, 1- \frac{ E_n+\bar{E}_n}{ T\eta P_n|h_{n,0}|^2}, \frac{T\eta P_n|h_{n,0}|^2-  \bar{E}_n}{T\eta P_n|h_{n,0}|^2 + TP_{\max}}\right\} \label{2st:8}.
  \end{alignat}
\end{problem} 

Denote $x^*$ as the optimal solution of the   maximization problem $\underset{x\geq 0}{\max}\text{ } \tilde{f}^*_0(x)$. As indicated by Proposition \ref{proposition2}, $\tilde{f}^*_0(x)$ is a concave function of $x$ for $x\geq 0$, and hence there exists a single maximum, i.e., $x^*$.  However, because $x^*$ is obtained by discarding the constraints, \eqref{1st:8} and \eqref{2st:8}, it is possible that $x^*$ is not a feasible solution of   problem \ref{pb:8}. 

Therefore, the lemma is proved in the following three steps. In the first step, we   analyze the properties  of $\tilde{f}^*_0(x)$ at  the boundary values given by   \eqref{1st:8} and \eqref{2st:8}. In the second step, a closed-form expression for $x^*$ is developed, where the conclusion about the comparison   between $x^*$ and one of the boundary values, $1-  \frac{\bar{E}_n}{ T\eta P_n|h_{n,0}|^2}$, is crucial to avoid the ambiguity caused by the two branches of the   Lambert W function. In the third step, a closed-form solution for problem \ref{pb:8} is obtained by using the obtained closed-form $x^*$.

\subsubsection{Analysis of the properties  of  $\tilde{f}^*_0(x)$}
The aim of this subsection is to show that $\tilde{f}^*_0(x)$ is non-decreasing at $x=0$ and non-increasing at $1-  \frac{\bar{E}_n}{ T\eta P_n|h_{n,0}|^2}$,    the   boundary values shown in  \eqref{1st:8} and \eqref{2st:8}, which can be proved by showing the following
\begin{align}\label{bounds proofx}
f'(0) \geq 0\quad \& \quad   f'\left(1-  \frac{\bar{E}_n}{ T\eta P_n|h_{n,0}|^2}\right)\leq 0,
\end{align}
where $f'(x) \triangleq \frac{d\tilde{f}_0^*(x)}{dx}$ denotes the first order derivative of $\tilde{f}_0^*(x)$. 
We note that   \eqref{bounds proofx} requires the implicit assumption that both $0$ and  $1-  \frac{\bar{E}_n}{ T\eta P_n|h_{n,0}|^2}$ are in the domain of the objective function, which can be easily shown by using \eqref{uob7y}.

In order to show $f'(0) \geq 0$, assume  $x=\epsilon\rightarrow 0$. Thus,  by using \eqref{first orderxx}, $f'(\epsilon)  $ can be expressed as follows:
\begin{align}  \nonumber
f'(\epsilon) =  &   \ln\left( 1-\kappa_1+  \frac{ \kappa_1-\kappa_2}{\epsilon} \right)- \frac{(\kappa_1-\kappa_2)}{ \epsilon(1-\kappa_1)+    \kappa_1-\kappa_2  } \\
\underset{\epsilon\rightarrow 0}{\rightarrow}&\ln\left( 1-\kappa_1+  \frac{ \kappa_1-\kappa_2}{\epsilon} \right)- 1.
\end{align}

Note that $\kappa_1-\kappa_2\geq 0$, which can be proved  as follows:
\begin{align}
\kappa_1-\kappa_2 =& \frac{ \eta P_n|h_{n,0}|^2 |g_{0,n}|^2}{1 + P_n|h_n|^2}- \frac{ \bar{E}_n|g_{0,n}|^2}{T(1 + P_n|h_n|^2)}\\\nonumber =&\frac{ \left(T\eta P_n|h_{n,0}|^2  - \bar{E}_n\right) |g_{0,n}|^2}{T(1 + P_n|h_n|^2)}    \geq 0,
\end{align}
which is due to the fact  that  $\bar{E}_n\leq T\eta P_n|h_{n,0}|^2  $ as shown in \eqref{bound for En}.  As a result, $\frac{ \kappa_1-\kappa_2}{\epsilon}$ approaches positive infinity for $\epsilon\rightarrow 0$, where the trivial case $\bar{E}_n= T\eta P_n|h_{n,0}|^2  $ is not considered as discussed in Remark 5. 
Therefore, for $\epsilon\rightarrow 0$, $f'(\epsilon)$ can be approximated as follows: \begin{align}\label{bound1}
f'(\epsilon) \underset{\epsilon\rightarrow 0}{\rightarrow} &     \ln\left( 1-\kappa_1+  \frac{ \kappa_1-\kappa_2}{\epsilon} \right)- 1  \geq 0 .
\end{align}


In order to show $f'\left(1-  \frac{\bar{E}_n}{ T\eta P_n|h_{n,0}|^2}\right)\leq 0$, define  $\theta_1=1-  \frac{\bar{E}_n}{ T\eta P_n|h_{n,0}|^2}$ for simplicity   of notation.  
The first order derivative  $f'\left(\theta_1\right)$ is given by:  
\begin{align} 
f'\left(\theta_1\right)  =&\ln\left( 1-\kappa_1+  \frac{ \kappa_1-\kappa_2}{\theta_1} \right)- \frac{(\kappa_1-\kappa_2)}{ \theta_1(1-\kappa_1)+    \kappa_1-\kappa_2  }  . 
\end{align}
The following identity holds:
\begin{align}\nonumber
 \frac{ \kappa_1-\kappa_2}{\theta_1}  = \frac{ \left(T\eta P_n|h_{n,0}|^2  - \bar{E}_n\right) |g_{0,n}|^2}{T(1 + P_n|h_n|^2)}    \frac{ T\eta P_n|h_{n,0}|^2}{T\eta P_n|h_{n,0}|^2-\bar{E}_n}= \frac{\eta P_n|h_{n,0}|^2|g_{0,n}|^2}{(1 + P_n|h_n|^2)}   =\kappa_1,
\end{align}
which means 
\begin{align} 
f'\left(\theta_1\right)  = - \frac{(\kappa_1-\kappa_2)}{ \theta_1(1-\kappa_1)+    \kappa_1-\kappa_2  } . 
\end{align}
$f'\left(\theta_1\right) $ can be further rewritten  as follows:
\begin{align} 
f'\left(\theta_1\right)  = - \frac{\frac{(\kappa_1-\kappa_2)}{\theta_1}}{  (1-\kappa_1)+ \frac{   \kappa_1-\kappa_2}{\theta_1}  } =  - \frac{\kappa_1}{  (1-\kappa_1)+ \kappa_1 }=-\kappa_1 \leq 0,
\end{align}
since $\kappa_1$ is non-negative. Therefore, the inequality in \eqref{bounds proofx} is proved. 

By using the concavity of the objective function and also \eqref{bounds proofx}, it is straightforward to show that  $x^*$ is bounded as follows: 
\begin{align}\label{boundall}
0\leq x^*\leq \theta_1.
\end{align} 
The bounds in \eqref{boundall} will be useful to solve an ambiguity issue caused by  the Lambert W function, as shown in the following subsection.  

\subsubsection{Finding a closed-form expression for $x^*$}
The fact that  the objective function $\tilde{f}_0^*(x)$ is concave means that $x^*$ is the root of the following equation:  
\begin{align}
   \ln\left( 1-\kappa_1+  \frac{ \kappa_1-\kappa_2}{x^*} \right)- \frac{(\kappa_1-\kappa_2)}{ x^*(1-\kappa_1)+    \kappa_1-\kappa_2  } =0,
\end{align}
which can be rewritten as follows:
\begin{align}
   \ln\left( 1-\kappa_1+  \frac{ \kappa_1-\kappa_2}{x^*} \right)- \frac{\frac{(\kappa_1-\kappa_2)}{x^*}}{ (1-\kappa_1)+    \frac{\kappa_1-\kappa_2 }{x^*} } =0.
\end{align}
Define $y^*=1-\kappa_1+  \frac{ \kappa_1-\kappa_2}{x^*}$, which leads to:  
\begin{align}
   \ln\left( y^* \right)- \frac{(y^*-1+\kappa_1)}{ y^*} =0,
\end{align}
or equivalently 
\begin{align}
 y^*  \ln\left( y^* \right)-  y^*=\kappa_1-1. 
\end{align}
In order to apply the the Lambert W function, define $y^*=e^{z^*}$, which results in the following equation:
\begin{align}
z^*e^{z^*}   -  e^{z^*}=\kappa_1-1
\end{align}
or equivalently 
\begin{align}\label{zzzc}
(z^*-1)e^{z^*-1}   =e^{-1}(\kappa_1-1).
\end{align}
The form of \eqref{zzzc} makes the application of the Lambert W function possible. Recall that there is a single real-valued solution for the equation $we^{w}=z$, if $z\geq 0$, which  is given by $w=W_0(z)$, i.e., the principal branch of the Lambert W function. However, if $z$ is negative, there are two real-valued solutions, denoted by $w=W_0(z)$ and $w=W_{-1}(z)$, respectively. 

For the considered problem,    $\kappa_1-1$ is not always positive since
\begin{align}\label{mu11}
\kappa_1-1 =\frac{ \eta P_n|h_{n,0}|^2 |g_{0,n}|^2}{1 + P_n|h_n|^2}-1
\end{align}
can be negative if $|h_n|^2$ is much larger than $|h_{n,0}|^2$. As a result, for the case $\kappa_1-1<0$,   there are two real-valued solutions corresponding to the two branches of the Lambert W function as follows:
\begin{align}
z^*-1 =& W_0\left(e^{-1}(\kappa_1-1)\right), 
\\ 
 \quad z^*-1 =& W_{-1}\left(e^{-1}(\kappa_1-1)\right).
\end{align}
With the two possible choices for $z^*$, there will be also two choices for $x^*$, which makes the final expression of $\alpha^*(\bar{E})_n$ very complicated. Fortunately,  the solution corresponding to $z^*-1 = W_{-1}\left(e^{-1}(\kappa_1-1)\right)$ can be discarded, as proved in the following. 

  In the following, assume $\kappa_1-1<0$, and proof by contradiction will be used to show that $z^*-1 = W_{-1}\left(e^{-1}(\kappa_1-1)\right)$ is not a feasible solution.  Note that, according to \eqref{mu11},  $\frac{ \eta P_n|h_{n,0}|^2 |g_{0,n}|^2}{1 + P_n|h_n|^2}$ is strictly positive, and hence $\kappa_1-1$ is strictly larger than $-1$, i.e. $-1<\kappa_1-1<0$,  According to the property of the Lambert W function, $w=W_{-1}(z)< -1$, for $-1<z<0$. Therefore,   $ z^*-1 = W_{-1}\left(e^{-1}(\kappa_1-1)\right)<-1$,which means 
\begin{align}
z^* < 0 \Longrightarrow y^*=e^{z^*}< 1.
\end{align}
By using the relationship between $y^*$ and $x^*$, $x^*$ needs to satisfy the following inequality:
\begin{align}
1-\kappa_1+  \frac{ \kappa_1-\kappa_2}{x^*}< 1\Longrightarrow     \kappa_1(1-x^*)< \kappa_2 ,
\end{align}
where the fact that $x^*\geq 0$ is used. 
By substituting  the expressions of $\kappa_1$ and $\kappa_2$ into the above inequality, $x^*$ needs to satisfy the following inequality:
\begin{align}
\frac{ \eta P_n|h_{n,0}|^2 |g_{0,n}|^2}{1 + P_n|h_n|^2}(1-x^*)< \frac{ \bar{E}_n|g_{0,n}|^2}{T(1 + P_n|h_n|^2)}  ,
\end{align}
which can be further rewritten     as follows:
\begin{align}\label{wrong conclusion}
T\eta P_n|h_{n,0}|^2   (1-x^*)<   \bar{E}_n   .
\end{align}
In the following, we  show that the condition  in \eqref{wrong conclusion} cannot be met. According to our previous analysis, $f'(\theta_1)\leq 0$, which implies  the following: 
\begin{align}
x^*\leq \theta_1\triangleq 1-  \frac{\bar{E}_n}{ T\eta P_n|h_{n,0}|^2}. 
\end{align} 
With some algebraic manipulations, the inequality that $x^*$ needs to satisfy can be expressed as follows:
\begin{align} 
   (1- x^*) T\eta P_n|h_{n,0}|^2\geq  \bar{E}_n,
\end{align} 
which contradicts the condition  in \eqref{wrong conclusion}.  Therefore, the solution $z^*-1 = W_{-1}\left(e^{-1}(\kappa_1-1)\right)$ can be discarded, and $z^*-1 = W_0\left(e^{-1}(\kappa_1-1)\right)$ is the   solution which should be used.

Therefore, with $z^*-1 = W_0\left(e^{-1}(\kappa_1-1)\right)$,  $y^*$ can be written as follows:  
\begin{align}
  y^*=e^{z^*}=e^{ W_0\left(e^{-1}(\kappa_1-1)\right)+1} ,
\end{align}
which means that a closed-form expression for $x^*$ is obtained as follows: 
\begin{align}
 x^* = \frac{\kappa_1-\kappa_2}{e^{ W_0\left(e^{-1}(\kappa_1-1)\right)+1}  - 1+\kappa_1} .
\end{align}

\subsubsection{Finding the optimal solution for problem \ref{pb:8}}
Although  $x^*$   yields the maximal value for the objective function,   $x^*$ is not necessarily the optimal solution of problem \ref{pb:8}, since $x^*$ might violate one of the constraints of problem \ref{pb:8}.  To facilitate the discussions, the constraints of problem \ref{pb:8} are rewritten as follows:
\begin{align}
   \max\left\{0, \theta_0\right\}\leq  \alpha_n \leq \min\left\{1, \theta_1\right\}  ,
\end{align}
which defines the feasible set of $\alpha_n$, 
where $\theta_0$ is defined in Lemma \ref{lemma1}. While the relationship between $0$, $\theta_1$ and $x^*$ is fixed as shown in \eqref{boundall},  the relationship between $x^*$, $1$, and $\theta_0$ is not fixed, and there are three possible cases which yield three different solutions:
\begin{itemize}
\item Case 1: When $\max\{0, \theta_0\}\leq x^*$ and $x^*\leq \min\{1,\theta_1\}$,   $x^*$ is inside of the feasible set of   problem \ref{pb:8} as shown in Fig. \ref{figx00}. Hence, $x^*$ is the optimal solution of the problem.
\item Case 2: When $ x^*<\theta_0$, the feasible set of $\alpha_n$ is $[\theta_0, \min\{1,\theta_1\}]$ and $x^*$ is at the left-hand side of the feasible set as shown in Fig. \ref{figx00}. The function is monotonically decreasing over the feasible set, and hence  $\theta_0$ is the optimal solution of the problem.
\item Case 3: When $ x^*>1$, the feasible set of $\alpha_n$ becomes $[\max\{0, \theta_0\},1]$ and  $x^*$ is at the right-hand side of the feasible set as shown in the figure. The function is monotonically  increasing over the feasible set, and hence  $1$ is the optimal solution of the problem. 
\end{itemize}
In summary, the optimal solution for problem \ref{pb:8} is given by
\begin{align}
\alpha_n^*(\bar{E}_n) =  \min\left\{1, \max\left\{ x^*, \theta_0\right\} \right\},
\end{align}
and the lemma is proved. 

\begin{figure}[t]\centering \vspace{-2em}
    \epsfig{file=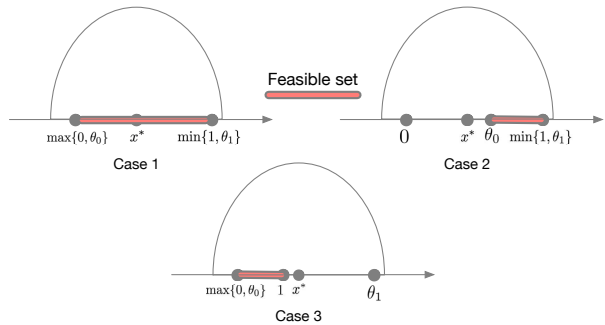, width=0.48\textwidth, clip=}\vspace{-2em}
\caption{ Illustration for three possible cases for $x^*$}\label{figx00}\vspace{-2em}
\end{figure} \vspace{-2em}
\linespread{1.25}
     \bibliographystyle{IEEEtran}
\bibliography{IEEEfull,trasfer}

   \end{document}